\documentclass[11pt]{article}
\usepackage[margin=1in]{geometry}
\usepackage[final]{microtype}
\usepackage{epsfig}
\usepackage{graphics}
\usepackage{latexsym}
\usepackage{amsmath}
\usepackage{amsfonts}
\usepackage{amssymb}
\usepackage{mathrsfs}
\usepackage[dvipsnames]{xcolor}
\usepackage{amsthm}
\usepackage{xspace}
\usepackage{epstopdf}
\usepackage{float}
\usepackage{hyperref}
\usepackage{pgfplots}
\pgfplotsset{compat=1.18}
\usepackage{longtable}

\numberwithin{equation}{section}
\usepackage[ruled,vlined]{algorithm2e}
\SetArgSty{textrm}

\usepackage[english]{babel}
\usepackage[nottoc]{tocbibind}

\usepackage{caption}
\usepackage{subcaption}
\usepackage{pifont}
\usepackage{booktabs}
\usepackage{multirow}

\theoremstyle{plain}
\newtheorem{theorem}{Theorem}

\newtheorem{corollary}{Corollary}

\newtheorem{lemma}{Lemma}

\newtheorem{mechanism}{Mechanism}
\theoremstyle{definition}

\theoremstyle{remark}

\makeatletter
\@addtoreset{equation}{section}
\def\section{\@startsection {section}{1}{\z@}{-3.5ex plus -1ex minus
 -.2ex}{2.3ex plus .2ex}{\large\bf}}
\makeatother

\def\bfm#1{\mbox{\boldmath$#1$}}

\def\0{\bfm 0}

\DeclareMathAlphabet{\mathpzc}{OT1}{pzc}{m}{it}

\newcounter{my}

\newcounter{my2}

\newcounter{my3}

\newcounter{my4}

\newcounter{my5}

\newcounter{my6}

\allowdisplaybreaks

\DeclareUnicodeCharacter{2011}{\mbox{-}}
\begin{document}

\title{Strategyproof Mechanisms for Connecting Impassable Regions}
\author{Hau Chan$^{1}$\quad Jianan Lin$^{2}$\quad Chenhao Wang $^{3,4}$\\[0.75em]
$1$ University of Nebraska-Lincoln\\
$2$ Rensselaer Polytechnic Institute\\
$3$ Beijing Normal University-Zhuhai\\
$4$ Beijing Normal-Hong Kong Baptist University
}
\date{}
\maketitle

\begin{abstract}
We study strategyproof mechanisms for building a pathway between two regions of a line segment separated by an obstacle. Each of the $n$ agents has a private location within its region and may use either its original route to a facility or the new pathway, whose traversal cost is a fraction $k\in[0,1)$ of its length. We seek strategyproof (SP) and group-strategyproof (GSP) mechanisms that approximately minimize maximum cost or social cost. After characterizing optimal pathways for both objectives, we establish a tight deterministic maximum-cost approximation ratio of $\frac{2}{1+k}$ and a deterministic social-cost upper bound of $\frac{n}{1+k(n-1)}$, together with complementary lower bounds. Both upper bounds are achieved by GSP mechanisms. We then study randomized mechanisms under strategyproofness in expectation. A power-proportional mechanism achieves a social-cost approximation ratio at most $5$, independent of $n$ and $k$, with a tight guarantee of $3$ for this mechanism when $k=0$. We prove randomized lower bounds of $\frac{3+2k}{2+3k}$ for maximum cost and $\max\big\{1,\frac{285}{263+385k}\big\}$ for social cost, the latter for $n\ge7$. Finally, we improve several bounds for the real-line pathway model of [Chan and Wang, AAMAS 2023]. Our deterministic maximum-cost lower bound of $2$ matches the upper bound obtainable from [Qin, Fang, and Liu, COCOA 2024]. We strengthen the deterministic social-cost lower bound from $\frac32$ to $2$ under SP and to $\max\{2,n-1\}$ under GSP. For randomized social cost, we sharpen the guarantee of Chan and Wang's proportional mechanism from $6$ to $3$ and raise their lower bound from $1.02$ to $\frac{285}{263}\approx1.08365$ for $n\ge7$.
\end{abstract}

\section{Introduction}\label{sec:intro}

Road infrastructure supports everyday travel and access to essential services. 
It encompasses roads, bridges, sidewalks, and multi-use paths for pedestrians and different kinds of vehicles \cite{garber2009traffic}. 
It enables individuals in different regions to travel safely and efficiently between locations for daily activities (e.g., going to work, school, hospital, and shopping) \cite{amekudzi2007transportation,forkenbrock1990economic,narayanaswami2017urban}. 

While modern road infrastructure is designed with connectivity in mind, it can sometimes contain impassable regions and be disconnected due to inherent physical limitations \cite{jenelius2012road}.
For example, in areas prone to natural disasters (e.g., floods and earthquakes) that can damage roads, separate disconnected roadways are often built between such areas instead of continuous roadways \cite{kermanshah2016geographical}. 
In areas that contain protected species or have high wildlife movement, disconnected roadways are built to avoid these areas \cite{forman1998roads}. 
Similarly, large water bodies, mountainous terrain, and military zones can create natural or deliberate separations in road networks \cite{rodrigue2020geography,reggiani2015transport}.
As a result, due to the inability to travel through impassable regions (separated by such areas), individuals often confine their daily activities to their own regions \cite{jenelius2009network}.

When regions are disconnected, individuals face a fundamental trade-off in accessing essential services and amenities (e.g., schools and shopping centers): they must either rely solely on facilities within their own region, potentially at greater distances or lower quality, or forgo access to closer or superior facilities in neighboring regions that remain unreachable due to physical barriers \cite{syed2013traveling,rodrigue2020geography}.
By establishing a pathway that connects previously isolated regions, a social planner can expand individuals' accessible options, enabling them to choose facilities based on proximity and preference rather than regional boundaries, thereby improving overall welfare and resource utilization across the broader network \cite{talen1998assessing}.

%\bluecomment{In reality, travelers typically face multiple potential destinations and make choices based on relative travel cost and accessibility, rather than being forced to go to one predetermined point. Transportation planning and travel demand models explicitly incorporate such behavior through destination choice frameworks, where the utility of a destination depends on travel impedance and accessibility to opportunities beyond a single fixed destination \cite{ben1985discrete,bhat1998disaggregate}. Similarly, measures of spatial accessibility—such as gravity-based accessibility \cite{geurs2004accessibility} and two-step floating catchment area methods \cite{luo2003measures}—evaluate how the set of reachable facilities within travel cost thresholds affects access and welfare outcomes rather than just proximity to one outlet \cite{park2021review}. In our setting, agents are likewise allowed to use the introduced pathway only if it reduces their individual travel cost relative to local alternatives, better reflecting the behavioral insight that infrastructure improvements expand choice sets and influence travel decisions only when they afford lower cost options.}

\begin{figure}[H]
\vspace{-2em}
\centering
\includegraphics[width=10cm]{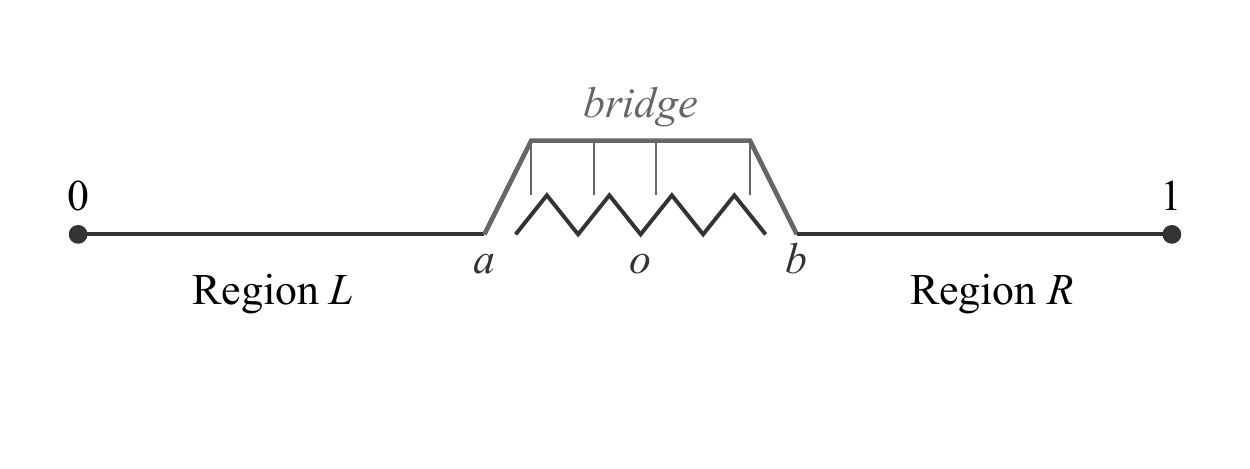}
\vspace{-3em}
\caption{A schematic diagram for building a pathway.}
\label{fig:obs}
\vspace{-1em}
\end{figure}

\paragraph*{Existing Efforts to Connect Impassable Regions} 

To connect impassable regions and potentially improve the reachability of individuals beyond their own regions, existing approaches \cite{chen2007network,cai2024spatial} focus on designing alternative pathways or bridges to connect the two regions. 
In most cases, these pathways serve as detours and become an integral part of the road infrastructure that connects regions over the long term to enhance existing connectivity. 
Figure \ref{fig:obs} illustrates a simplified example of building a pathway to connect two impassable regions on a bounded road segment divided by an obstacle. 
The two impassable regions (Region L and Region R) were originally disconnected due to physical limitations, which created an obstacle $o$, making direct travel between them impossible. 
Therefore, the social planner aims to build a new pathway with endpoints $(a, b)$ to connect the two regions.

Existing algorithmic and mechanism design studies focus on identifying the appropriate endpoints $(a, b)$ for the pathway to minimize specific cost objectives based on points in the regions. From the algorithmic perspective, related studies \cite{leizhen1999computing,Kim1998Linear,kim2001computing,tan2000optimal,tan2002finding} involve developing efficient algorithms for constructing pathways between convex polygons that often minimize the maximum distance between all points in the regions. 
To incorporate agents' preferences into pathway construction, recent mechanism design studies \cite{Chan023,chan2025mechanism,qin2024mechanism} use agents' starting locations to determine the pathway endpoints. 
Because starting locations are private information and agents can misreport them, these studies focus on designing strategyproof mechanisms to elicit agent starting positions truthfully and determine endpoints that approximately optimize cost objectives based on minimizing the distances of each agent's starting position from one region to another region using the pathway.  
In this paper, to account for agents' input, we build upon existing mechanism design studies above for building pathways that focus on improving the connectivity of agents, taking into account their distances from their starting positions to their own regions and to other regions via the pathway.

\subsection{Our Contributions} 
We consider the mechanism design perspective of constructing (approximately) optimal pathways between two regions that are separated by obstacles, with a focus on improving the connectivity of the agents to both regions. 
Following the setting of \cite{chan2025mechanism}, we consider a line segment, represented by the interval $[0,1]$. 
An obstacle $o$ divides the line segment into two disjoint regions with corresponding roadways. 
Naturally, agents are partitioned into sets $N_1$ and $N_2$ based on their (starting) locations relative to the obstacle. 
Agents in $N_1$ are positioned on the left-hand side of the obstacle ($x_i \in [0, o)$), while agents in $N_2$ are on the right-hand side ($x_i \in (o,1]$).

Our goal is to elicit agents' locations truthfully and construct a pathway $(a,b)$ connecting the two regions, with $a \in [0,o)$ and $b \in (o,1]$. 
For an agent at $x_i \in [0,o)$, the cost is defined as the minimum of $|x_i-a| + k(b-a) + 1-b$ and $|x_i-0|$, where $k \ge 0$ is a scaling factor, reflecting the distance to one point in $\{0, 1\}$ either within their own region or across the edge $(a,b)$. Similarly, for an agent at $x_i \in (o,1]$, the cost is the minimum of $|x_i-b| + k(b-a) + a$ and $|1-x_i|$. 
%That is, an agent will use the pathway to reach the other region if it is beneficial to do so in terms of distances. 
Our cost model incorporates well-known behavior from destination choice frameworks \cite{ben1985discrete,bhat1998disaggregate}, where an agent determines regions to access based on relative travel cost and accessibility. Therefore, an agent will use the pathway to reach the other region if it is more cost-effective (e.g., reducing their travel cost relative to the local alternative in their own region).
%\bluecomment{
%In reality, travelers typically face multiple potential destinations and make choices based on relative travel cost and accessibility, rather than being forced to go to one predetermined point. 
%Transportation planning and travel demand models explicitly incorporate such behavior through destination choice frameworks, where the utility of a destination depends on travel impedance and accessibility to opportunities beyond a single fixed destination \cite{ben1985discrete,bhat1998disaggregate}. 
%Similarly, measures of spatial accessibility—such as gravity-based accessibility \cite{geurs2004accessibility} and two-step floating catchment area methods \cite{luo2003measures}—evaluate how the set of reachable facilities within travel cost thresholds affects access and welfare outcomes rather than just proximity to one outlet \cite{park2021review}. 
%In our setting, agents are likewise allowed to use the introduced pathway only if it reduces their individual travel cost relative to local alternatives, better reflecting the behavioral insight that infrastructure improvements expand choice sets and influence travel decisions only when they afford lower cost options.}
We evaluate pathways under two objectives: \emph{maximum cost}, the largest individual cost, and \emph{social cost}, the sum of all agents' costs.

We focus on $0\le k<1$, and consider both deterministic mechanisms and randomized mechanisms that are strategyproof in expectation. The latter are evaluated by expected social cost or expected maximum cost.

\begin{figure}[htbp]
\centering
\includegraphics[width=\textwidth]{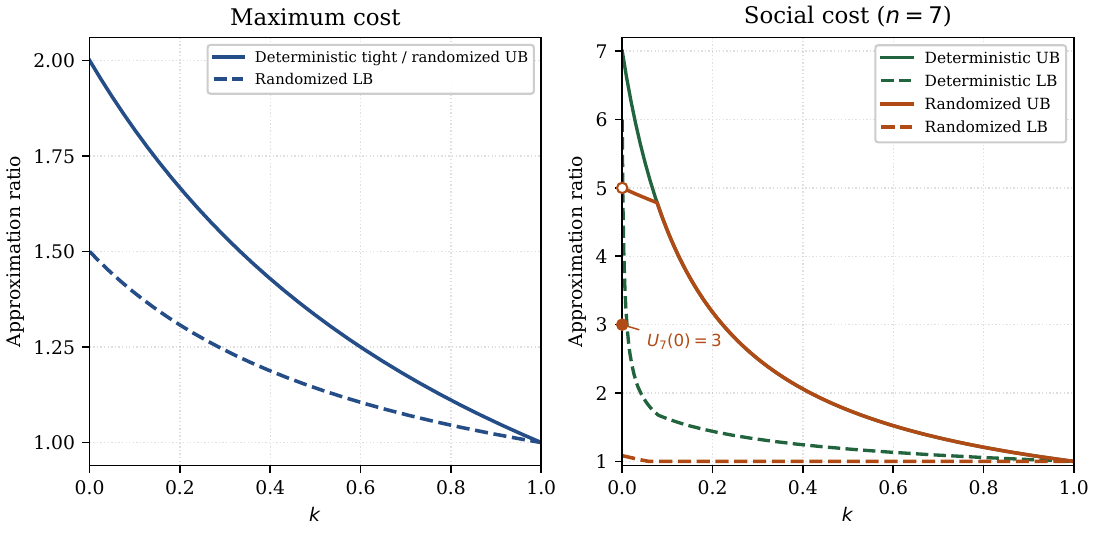}
\caption{Approximation bounds for deterministic and randomized SP mechanisms ($n=7$).}
% \caption{Deterministic and randomized bounds for the obstacle model with $n=7$. Solid and dashed curves denote upper and lower bounds, respectively. The deterministic maximum-cost bound is also the available randomized upper bound. At $k=0$, the randomized social-cost upper bound is $3$ (filled point); the open point at $5$ is the right-hand limit of the bound for $k>0$. This discontinuity concerns the proved guarantee, not the mechanism. The social-cost upper-bound curves coincide where the randomized bound selects \textsc{TwoExtreme}. Values at $k=1$ are limits.}
\label{fig:result}
\end{figure}

Our contributions are as follows.
\begin{itemize}

\item We characterize optimal pathways for both objectives. For every fixed pathway, agents in at most one region can strictly improve over their original route. Consequently, an optimal pathway can be chosen with one endpoint at an original facility. We give optimal mechanisms and show that they are not strategyproof.

\item For deterministic mechanisms, \textsc{CriticalExtreme} is GSP and achieves the maximum-cost ratio $\frac{2}{1+k}$, which matches our lower bound for all deterministic SP mechanisms. For social cost, the GSP \textsc{TwoExtreme} mechanism achieves $\frac{n}{1+k(n-1)}$. We prove complementary lower bounds for deterministic SP mechanisms, including $\max\{2,n-1\}$ at $k=0$ and the bounds $\Lambda_m(k)$ in Theorem~\ref{thm:sc-lb} for admissible active populations $m\le n$.

\item For randomized mechanisms, we obtain a social-cost guarantee independent of the population size. Our \textsc{PowerProportional} mechanism samples agents according to a power of their maximum possible cost reduction. It is SP in expectation and achieves a social-cost ratio at most $1+C(k)\le5$ for $0<k<1$, where $C(k)$ is defined in \eqref{eq:random-sc-C}. At $k=0$, it reduces to proportional sampling, for which we prove a tight guarantee of $3$. Choosing between this mechanism and \textsc{TwoExtreme} using only the public parameters gives the bound $U_n(k)$ in \eqref{eq:random-sc-upper}. We also prove randomized lower bounds of $\frac{3+2k}{2+3k}$ for maximum cost and $\max\big\{1,\frac{285}{263+385k}\big\}$ for social cost, with $n\ge7$ for the latter. Thus, at $k=0$, the randomized maximum-cost bounds are $\frac32$ and $2$, while randomization reduces the dependence of the social-cost guarantee on $n$ from linear to constant.

\item We improve several bounds of Chan and Wang~\cite{Chan023} for the real-line pathway model, allowing arbitrary endpoints and unrestricted reports. For deterministic maximum cost, we raise their lower bound from $\frac32$ to $2$. The upper bound of $2$ in \cite{cocoa2024mechanism} carries over by connecting its selected facility location to the origin, giving a tight deterministic bound. For deterministic social cost, we raise the lower bound from $\frac32$ to $2$ under SP and to $\max\{2,n-1\}$ under GSP. For randomized mechanisms, we retain Chan and Wang's maximum-cost lower bound of $\frac32$, sharpen their proportional mechanism's social-cost guarantee from $6$ to $3$, and raise their randomized social-cost lower bound from $1.02$ to $\frac{285}{263}\approx1.08365$ for $n\ge7$. The factor $3$ is tight for the proportional mechanism; the optimal randomized SP ratio remains open.
\end{itemize}

\subsection{Related Work}

We review existing algorithmic and mechanism design studies on building pathways to minimize specific cost objectives based on points in the regions. 
We note that there is a rich line of work on mechanism design for facility location problems (e.g., see \cite{Moulin1980,procaccia2013approximate,lu10mechanism,alon2010strategyproof,feldman2013strategyproof,lin2020nearly}) aiming to design strategyproof mechanisms to elicit agents' preferred facility locations truthfully and determine facility locations that optimize certain cost objectives. 
In addition to the classical model,
other notable extensions include 
 facility preferences \cite{fong2018facility,PaoloCarmine2016}, and different cost functions \cite{feldman2013strategyproof,fotakis2013strategyproof}. For an overview, see \cite{chan2021mechanismsurvey}.
Our randomized mechanisms use related proportional-sampling ideas, while positive traversal costs and region-constrained reports require separate incentive and approximation analyses. 

\paragraph*{Algorithmic Perspectives of Building Pathways.} 
Optimization literature has addressed the problem of constructing an optimal pathway to connect disconnected regions. 
 \cite{leizhen1999computing} formulated the problem of adding a line segment to connect two disjoint convex polygonal regions, minimizing the longest path from a point in one polygon to a point in the other via the line segment. They proposed an $O(n^2 \log n)$-time algorithm, where $n$ denotes the maximum number of extreme points. 
Subsequently, \cite{bhattacharya2001computing} developed a linear-time algorithm improving upon this approach. Other works include \cite{tan2000optimal,tan2002finding,kim2001computing,Kim1998Linear}.
% \cite{tan2000optimal} independently offered an alternative linear-time algorithm and further generalized it to an $O(n^2)$-time solution for connecting two convex polyhedra in space.  \cite{kim2001computing} provided algorithms for finding optimal pathways between two convex polygons, two simple non-convex polygons, and a convex plus a simple non-convex polygon, with respective complexities of $O(n)$, $O(n^2)$, and $O(n\log n)$. Later, \cite{tan2002finding} presented an $O(n\log^3 n)$-time algorithm for connecting two simple non-convex polygons.  \cite{Kim1998Linear} proposed a linear-time algorithm for computing an optimal pathway between two parallel lines separated by an obstacle. 
Importantly, all these works focus on connecting entire regions. 
In contrast, our study considers only a finite set of points—the agents’ starting locations—and adopts a mechanism design perspective, where locations are private.

\paragraph*{Mechanism Design Perspectives of Building Pathways.} 
The work most closely related to ours is \cite{chan2025mechanism}, which serves as the main source of inspiration for our study. 
The key difference between the two settings is that in our model, an agent cares about the shortest distance to the furthest points of the two regions, whereas in \cite{chan2025mechanism} an agent is only concerned about the distance to the furthest points of the other region.  
Our setting also shares similarities with \cite{qin2024mechanism}, where agents use a zero-cost perpendicular pathway to reach a facility in the other region. 
The main difference lies in the underlying metric: we consider a real line interrupted by an obstacle, while \cite{qin2024mechanism} dealt with two real lines connected by a perpendicular pathway. 
For $k=0$, a change of coordinates preserves costs between our setting and a restricted part of the real-line pathway model of \cite{Chan023}; the feasible outputs and reporting domains differ. Section~\ref{sec:improve} gives the precise correspondence and proves lower bounds for their model directly.
Finally, \cite{chan2024mechanism} explored the related problem of extending facility connectivity by adding an interval rather than a pathway to regions that are connected. 

% \paragraph{Mechanism Design for Facility Location.} The mechanism design problem we study falls under the umbrella of approximate mechanism design without money. Procaccia and Tennenholtz \cite{procaccia2013approximate} introduced this paradigm, emphasizing the inherent trade-offs between strategyproofness and approximation guarantees. For minimizing social cost, the median mechanism is both strategyproof and optimal. They further proved that the median mechanism achieves a tight 2-approximation for maximum cost. %, while a randomized mechanism, LRM, attains a 1.5-approximation, together with matching lower bounds. 
% %In two-dimensional Euclidean spaces, the coordinate-wise median (CM) mechanism provides a $\sqrt{2}$-approximation for social cost \cite{DBLP:conf/sagt/Meir19} and a 2-approximation for maximum cost \cite{DBLP:journals/scw/GoelH23}, both being optimal among deterministic strategyproof mechanisms.
% Other notable extensions include alternative metric spaces \cite{alon2010strategyproof,feldman2013strategyproof,lin2020nearly,lu10mechanism}, %chan2025prediction
%  facility preferences \cite{fong2018facility,PaoloCarmine2016}, and other cost functions \cite{feldman2013strategyproof,fotakis2013strategyproof}. See a survey in \cite{chan2021mechanismsurvey}.

\section{Preliminaries}\label{sec:model}

\paragraph{Models.}
Let $N=\{1,\ldots,n\}$ be the set of agents located on the interval $[0,1]$, and let $\mathbf x=(x_1,\ldots,x_n)$ represent their location profile. There is an obstacle at point $o \in (0,1)$, and we assume that no agent is positioned exactly at $o$. This obstacle partitions the agent set as $N=N_1\cup N_2$, with $N_1\cap N_2=\emptyset$, where $N_1 = \{i \in N \mid x_i < o\}$ contains the agents in the left region, and $N_2 = \{i \in N \mid x_i > o\}$ contains the agents in the right region. 
% \hau{this following sentence is not right for the model; because they can just stay at the same place; see the intro and argument above ... remove texts like that thorughout the paper} 

The obstacle prevents direct travel between the two regions. To enable travel between them, we aim to construct a new pathway/edge $(a,b)$ connecting the two regions, with $a \in [0,o)$ and $b \in (o,1]$. Traversing this pathway incurs a cost of $k(b-a)$, where $k$ is a nonnegative constant.

Region membership is public and remains fixed under misreports. Let $\mathcal X=\prod_{i\in N_1}[0,o)\times\prod_{i\in N_2}(o,1]$ denote the reporting domain, with coordinates indexed by agent identities, and let $\mathcal F_o=[0,o)\times(o,1]$ be the set of feasible pathways. A deterministic mechanism $f:\mathcal X\to\mathcal F_o$ maps the location profile $\mathbf x$ to a pathway $f(\mathbf x)$. Given a pathway $f(\mathbf x)=(a,b)$, the \emph{cost} of an agent $i \in N_1$ in the left region is the distance to the closer point in $\{0, 1\}$:  
$$cost(a,b,x_i) = \min(|x_i-a| + k(b-a) + (1-b), x_i).$$
Similarly, for an agent $i \in N_2$ in the right region, the cost is  
$$cost(a,b,x_i) = \min(|x_i-b| + k(b-a) + a, 1-x_i).$$
%We only focus on deterministic mechanisms in the paper.

% A randomized mechanism maps $\mathbb R^n$ to probability distributions over $\mathbb R^2$. If $f(\mathbf x) = P$ is such a distribution, the cost of agent $i$ is the expected cost  
% $$cost(P,x_i) = \mathbb E_{(a,b) \sim P}[cost(a,b,x_i)].$$

%\subsection{\bluecomment{Truthfulness}}
\paragraph{Strategyproof Mechanisms.}
A mechanism $f$ is \emph{strategyproof} (SP) if no agent can decrease their cost by misreporting their location within their own region. Formally, $f$ is strategyproof if for any $i \in N$, any profile $\mathbf x$, and any $x_i'$ in the same region as $x_i$,  
$$cost(f(x_i, \mathbf x_{-i}), x_i) \le cost(f(x_i', \mathbf x_{-i}), x_i),$$  
where $\mathbf x_{-i}$ denotes the locations of all agents except $i$.  
Moreover, $f$ is \emph{group strategyproof} (GSP) if no coalition of agents can simultaneously misreport so that every member of the group strictly benefits. That is, for any nonempty $S \subseteq N$, any profile $\mathbf x$, and any  $\mathbf x_S'$ with $x_j',x_j$ in the same region  $\forall j\in S$, there exists some $i \in S$ such that  
$$cost(f(\mathbf x), x_i) \le cost(f(\mathbf x_S', \mathbf x_{-S}), x_i).$$

It is clear that GSP implies SP. Another property lying between them is \emph{partial group strategyproofness}. %\cite{lu10mechanism}. 
A mechanism $f$ is partial GSP if no coalition of agents located at the same point can jointly misreport so that every member strictly benefits. Formally, for any profile $\mathbf x$, nonempty coalition $S \subseteq N$ with $x_i=x_0$ for all $i\in S$, and any alternative reports $\mathbf x_S'$ within the same region as $x_0$, it holds %there exists $i \in S$ such that  
$$cost(f(\mathbf x), x_0) \le cost(f(\mathbf x_S', \mathbf x_{-S}), x_0).$$

As in many variants of facility location (see, e.g., \cite{lu10mechanism}),  SP implies partial GSP. 

%Clearly, group strategyproofness implies partial group strategyproofness, which in turn implies strategyproofness. In our setting (and indeed in virtually all facility location problems) strategyproofness already yields partial group strategyproofness. We record this observation in the following lemma, which will be used repeatedly without further comment.

\begin{lemma}\label{lemma:partial-gsp}
    Every strategyproof mechanism is partial group strategyproof.
\end{lemma}

\begin{proof}
Consider a coalition $S$ of agents located at $x_0$, and let $\mathbf x_S'$ be the alternative reports. 
 % Suppose that some coalition $S$ of agents located at $x_0$ jointly misreports to $\mathbf x_S'$ so that every member strictly benefits. 
 Move the agents in $S$ to their respective reports in $\mathbf x_S'$ one by one. At each step, the agent whose report is changed is still located at $x_0$. Strategyproofness therefore ensures that the cost evaluated at $x_0$ cannot decrease after that step. Hence,  the coalition $S$ cannot decrease their cost if they misreport $\mathbf x_S'$ simultaneously. % after all deviations, the cost of an agent at $x_0$ cannot decrease. contradicting the assumed strict improvement.
\end{proof}

% \begin{proof}
% Consider a coalition $S$ of agents located at $x_0$, and let $\mathbf x_S'$ be the alternative reports. 
%  % Suppose that some coalition $S$ of agents located at $x_0$ jointly misreports to $\mathbf x_S'$ so that every member strictly benefits. 
%  Move the agents in $S$ to their respective reports in $\mathbf x_S'$ one by one. Strategyproofness ensures that after each single-agent deviation the cost of an agent at $x_0$ cannot decrease. Hence,  the coalition $S$ cannot decrease their cost if they misreport $\mathbf x_S'$ simultaneously. % after all deviations, the cost of an agent at $x_0$ cannot decrease. contradicting the assumed strict improvement.
% \end{proof}

% A mechanism is \emph{anonymous} if its outcome is invariant under any permutation of agent identities, i.e., for any profile $\mathbf x$ and any permutation $\pi: N \to N$,  
% $$f(x_1,\ldots,x_n) = f(x_{\pi(1)},\ldots,x_{\pi(n)}).$$  
% Since non-anonymous mechanisms depend on agent identities and are less natural in this setting, we focus exclusively on anonymous mechanisms.

Our goal is to design SP/GSP mechanisms with performance guarantees for two objectives: minimizing the social cost and minimizing the maximum cost. For a given edge $(a,b)$ and location profile $\mathbf x$, the social cost is  
$SC(a,b,\mathbf x) = \sum_{i \in N} cost(a,b,x_i)$,
and the maximum cost is  
$MC(a,b,\mathbf x) = \max_{i \in N} cost(a,b,x_i)$.  
A mechanism $f$ is an $\alpha$-approximation ($\alpha \ge 1$) for objective $\Delta \in \{SC, MC\}$ if  
$$\Delta(f(\mathbf x), \mathbf x) \le \alpha \cdot \min_{(a,b) \in \mathcal F_o} \Delta(a,b,\mathbf x) \quad \text{for all } \mathbf x \in \mathcal X.$$

\paragraph{Randomized Mechanisms.}
A randomized mechanism maps each profile in $\mathcal X$ to a probability distribution over $\mathcal F_o$. It is \emph{strategyproof in expectation} if, for every agent $i$ and every same-region report $x_i'$, 
\[
 \mathbb E_{(a,b)\sim f(\mathbf x)}[cost(a,b,x_i)]
 \le
 \mathbb E_{(a,b)\sim f(x_i',\mathbf x_{-i})}[cost(a,b,x_i)].
\]
We use SP for this property when discussing randomized mechanisms. For $\Delta\in\{SC,MC\}$, its approximation ratio is at most $\rho$ if
\[
 \mathbb E_{(a,b)\sim f(\mathbf x)}[\Delta(a,b,\mathbf x)]
 \le\rho\,OPT_\Delta(\mathbf x)
 \qquad\text{for every }\mathbf x\in\mathcal X,
\]
where $OPT_\Delta$ minimizes $\Delta$ over feasible pathways. In particular, the maximum-cost objective is the expectation of the realized maximum cost. The proof of Lemma~\ref{lemma:partial-gsp} also applies to expected costs: moving the reports of colocated agents one at a time cannot decrease their common expected cost. We use this expectation version of partial group strategyproofness below. A randomized objective cannot improve on the unrestricted optimal objective, since it averages values of feasible deterministic pathways.

We study $0\le k<1$ throughout. Approximation guarantees are uniform over obstacle positions and profiles; a lower-bound construction may choose the obstacle as well as the agents' locations. All deviations keep each agent in its original region. The structural statements below concern individual pathways and deterministic mechanisms.

\paragraph{A Basic Lemma.}
The following lemma will be used repeatedly in our analysis.

\begin{lemma}
\label{lem:basic}
Consider any pathway $(a,b)$. If there exist agents whose cost decreases by using this pathway, then all such agents belong to the same region. In other words, at least one of the two regions has no agent who strictly benefits from using the pathway.
\end{lemma}

\begin{proof}
Assume for contradiction that there exist two agents $i \in N_1$ and $j \in N_2$ such that both benefit from the pathway $(a,b)$. By definition, this means $cost(a,b,x_i) < |x_i-0|$ and $cost(a,b,x_j)< |1-x_j|$, which indicates
\[
|x_i-a|+k(b-a)+(1-b) < x_i, \text{\quad and \quad}
(a-0)+k(b-a)+|b-x_j| < 1-x_j.
\]
However, adding the two left-hand sides gives
\[
\begin{aligned}
&|x_i-a|+k(b-a)+(1-b) + a+k(b-a)+|b-x_j| \ge |x_i-a|+(1-b) + a+|b-x_j|
\\
\ge ~& x_i-a+a+b-x_j+1-b =  x_i+1-x_j,
\end{aligned}
\]
 leading to a contradiction.
%Hence, it is impossible for agents in both regions to simultaneously benefit from the same pathway. 
Therefore, any pathway can only reduce the cost for agents in one region.
\end{proof}

% \begin{proof}
% Assume for contradiction that there exist two agents $i \in N_1$ and $j \in N_2$ such that both benefit from the pathway $(a,b)$. By definition, this means $cost(a,b,x_i) < |x_i-0|$ and $cost(a,b,x_j)< |1-x_j|$, which indicates
% \[
% |x_i-a|+k(b-a)+(1-b) < x_i, \text{\quad and \quad}
% (a-0)+k(b-a)+|b-x_j| < 1-x_j.
% \]
% However, adding the two left-hand sides gives
% \[
% \begin{aligned}
% &|x_i-a|+k(b-a)+(1-b) + a+k(b-a)+|b-x_j| \ge |x_i-a|+(1-b) + a+|b-x_j|
% \\
% \ge ~& x_i-a+a+b-x_j+1-b =  x_i+1-x_j,
% \end{aligned}
% \]
%  leading to a contradiction.
% %Hence, it is impossible for agents in both regions to simultaneously benefit from the same pathway. 
% Therefore, any pathway can only reduce the cost for agents in one region.
% \end{proof}

An agent who strictly benefits from a deviation must have a new cost strictly below the cost of taking the direct route, since that route was also available before the deviation. Lemma~\ref{lem:basic} therefore yields the following corollary.

\begin{corollary}\label{coro}
Let $f$ be a deterministic mechanism that is not group strategyproof. If a coalition $S\subseteq N$ can jointly misreport so that every member strictly benefits, then all agents in $S$ lie on the same side of the obstacle.
\end{corollary}

\section{Optimal Pathways}\label{sec:opt}

In this section, we focus on the optimal pathways, without imposing strategyproofness. We first give an important lemma.

Given a location profile $\mathbf x$, let $x_l = \min\{x_i \mid i\in N_1\}$, $x_r = \max\{x_i \mid i\in N_1\}$ be the two extreme locations in $N_1$; and $y_l = \min\{x_j \mid j\in N_2\}$, $y_r = \max\{x_j \mid j\in N_2\}$ be the two extreme locations in $N_2$. If $|N_1|=0$, set $x_l=x_r=0$; if $|N_2|=0$, set $y_l=y_r=1$.

\begin{lemma}\label{lem:opt-basic}
    For both maximum cost and social cost, there exists an optimal solution $(a^*, b^*)$ which satisfies $a^*=0$ or $b^*=1$.
\end{lemma}

\begin{proof}
By Lemma~\ref{lem:basic}, any agents who strictly benefit from a pathway belong to the same region. If they belong to $N_1$, extending $b$ to $1$ decreases their pathway costs by $(1-k)(1-b)$ and cannot increase any other agent's cost, since the direct route remains available. Symmetrically, if they belong to $N_2$, moving $a$ to $0$ cannot increase any cost. If no agent strictly benefits, either modification cannot increase any cost. Thus it suffices to consider pathways of the form $(a,1)$ or $(0,b)$.

For $(a,1)$, all agents in $N_2$ take the direct route at no greater cost. If $a>x_r$, moving $a$ to $x_r$ weakly decreases every left-side agent's cost, since their pathway costs have slope $1-k$ in $a$. Hence we may restrict $a$ to the compact interval $[0,x_r]$. Symmetrically, we may restrict $b$ to $[y_l,1]$ for pathways $(0,b)$. Both intervals consist of feasible endpoints, including when a region is empty. Continuity of the two objectives now guarantees an optimal solution in one of these two families.
\end{proof}

%------------------------------------------
\subsection{Maximum Cost}

The intuition behind the optimal solution for maximum cost is as follows. Consider the case $x_r+y_l\ge 1$ (as the other case $x_r+y_l< 1$ is symmetric). Without a pathway, the maximum cost is realized at $x_r$, so we need to reduce this cost. When $x_r>\frac{k}{1+k}$, we define $\delta_1$ as a threshold: an agent at $\delta_1$ (if one exists) can either go directly to 0 or use the pathway $(\frac{\delta_1+x_r}{2}, 1)$, so that this agent is indifferent between the two routes and has the same cost as the agent at $x_r$. Then we take the leftmost agent at or to the right of $\delta_1$, denoted $c_1$, and set $a^*$ as the midpoint between $c_1$ and $x_r$. More formally, we have the following optimal solution. 

\begin{mechanism}\label{mec:opt-mc}
Given profile $\mathbf{x}$, if $x_r + y_l \ge 1$, return $(a^*, 1)$ (Method 1);
    otherwise, return $(0, b^*)$ (Method 2).  

\textbf{Method 1.}  
If  $x_r\le \frac{k}{1+k}$, then set $a^* = x_r$.  
 Otherwise, let
\[
a^* = \frac{c_1+x_r}{2}, \quad 
c_1 = \min_{i\in N_1,\, x_i\ge \delta_1} x_i, \quad
\delta_1=\frac{2k+(1-k)x_r}{3+k}.
\]

\textbf{Method 2.}  
 If  $y_l \ge \frac{1}{1+k}$, then set $b^* = y_l$.  
 Otherwise, let
\[
b^* = \frac{c_2+y_l}{2}, \quad
c_2 = \max_{i\in N_2,\, x_i\le \delta_2} x_i, \quad
\delta_2 = \frac{2+(1-k)y_l}{3+k}.
\]
\end{mechanism}

\begin{theorem}
Mechanism \ref{mec:opt-mc} is optimal for maximum cost.
\end{theorem}

\begin{proof}
By symmetry, consider Method 1, where $x_r+y_l\ge 1$. The maximum cost without a pathway is $x_r$. Any pathway that does not strictly benefit an agent in $N_1$ leaves this value unchanged, whereas one that does can be extended to have $b=1$ without increasing any cost. Thus an optimal solution can be chosen with $b=1$, and agents in $N_2$ then have maximum cost $1-y_l$, with value $0$ if $N_2$ is empty.

If $x_r\le \frac{k}{1+k}$, every left-side agent satisfies $k(1-x_i)\ge x_i$. Since $|x_i-a|+k(1-a)$ is minimized at $a=x_i$, no such agent can improve on the direct route. Therefore the mechanism is optimal in this case.

Suppose now that $x_r>\frac{k}{1+k}$. For a left-side location $v\le x_r$, define
\[
D(v)=\frac{x_r-v}{2}+k\left(1-\frac{x_r+v}{2}\right).
\]
The maximum of the pathway costs of agents at $v$ and $x_r$ is minimized at $a=\frac{v+x_r}{2}$, with value $D(v)$. Consequently, for every $a$, their maximum actual cost is at least $\min\{v,D(v)\}$: if either agent takes the direct route, the maximum is at least $v$; otherwise, it is at least $D(v)$.

The threshold $\delta_1$ satisfies $D(\delta_1)=\delta_1$, and $D(v)-v$ is strictly decreasing. Moreover, $\delta_1<x_r$, so $c_1$ is well defined and $D(c_1)\le c_1$. Let
\[
v_0=\max\bigl(\{x_i:i\in N_1,\ x_i<\delta_1\}\cup\{0\}\bigr).
\]
The preceding two-agent bound applied to $c_1$ gives a lower bound $D(c_1)$ on the maximum cost of any $(a,1)$. If there is an agent at $v_0<\delta_1$, it also gives the lower bound $v_0$, since $D(v_0)>v_0$; if there is no such agent, this bound is simply $0$. Together with the right-side agents, these bounds yield
\[
MC(a,1,\mathbf x)\ge\max\{1-y_l,\ v_0,\ D(c_1)\}.
\]
At the mechanism's endpoint $a^*=\frac{c_1+x_r}{2}$, every agent with $x_i<\delta_1$ has cost at most $x_i\le v_0$. Every other left-side agent lies in $[c_1,x_r]$, so its pathway cost is at most $D(c_1)$. The mechanism therefore attains the displayed lower bound and is optimal.
\end{proof}

Mechanism~\ref{mec:opt-mc} is not strategyproof. Suppose there are two agents in $N_1$ and none in $N_2$, with $o=\frac{19}{20}$, $x_1=\frac35$, and $x_2=\frac45$. The mechanism returns $(\frac{7}{10},1)$. If agent $2$ reports $\frac9{10}$, the output becomes $(\frac34,1)$. Its true cost decreases from $\frac1{10}+\frac{3k}{10}$ to $\frac1{20}+\frac{k}{4}$, a strict reduction of $\frac{1+k}{20}$. Both reports lie in the left region, and the comparison holds for every $0\le k<1$. %We will design strategyproof mechanisms later. Next, we turn to the optimal mechanism for social cost.

\subsection{Social Cost}

The optimal solution for social cost is less intuitive than for maximum cost. 
It is obtained by enumerating all possible discrete candidate points and 
choosing the one that minimizes the objective. This is justified by the fact 
that the social cost function is piecewise linear, so an optimal endpoint can be chosen
from a finite set of candidate points.

\begin{mechanism}\label{mec:opt-sc}
Given profile $\mathbf{x}$, compare $(a^*,1)$ from Method~1 with $(0,b^*)$ 
from Method~2, and choose the better.

\textbf{Method 1.} 
 If  $x_r \le \frac{k}{1+k}$, then set $a^*=x_r$.
 Otherwise, choose $a^*$ to minimize social cost over the candidate set
    \[
    S_1 = \{x_i\!:\!i\in N_1\} \cup
          \left\{\frac{2x_i-k}{1-k} \!:\! i\in N_1, 0<\frac{2x_i-k}{1-k}<o\right\}.
    \]

\textbf{Method 2.} 
If  $y_l \ge \frac{1}{1+k}$, then set $b^*=y_l$.
Otherwise, choose $b^*$ to minimize social cost over the candidate set
    \[
    S_2 = \{x_i\!:\!i\in N_2\} \cup
          \left\{\frac{2x_i-1}{1-k} \!:\! i\in N_2, o<\frac{2x_i-1}{1-k}<1\right\}.
    \]
\end{mechanism}

\begin{theorem}\label{thm:sc-opt}
    Mechanism \ref{mec:opt-sc} is optimal for social cost.
\end{theorem}

\begin{proof}
By Lemma~\ref{lem:opt-basic}, it suffices to optimize over pathways $(a,1)$ and $(0,b)$ and compare their social costs. We first consider $(a,1)$. As in the proof of that lemma, we may restrict $a$ to $[0,x_r]$, and the total cost of agents in $N_2$ is constant.

Write $t=\frac{k}{1+k}$. If $x_r\le t$, then for every left-side agent the pathway cost is minimized at $a=x_i$ and even this minimum, $k(1-x_i)$, is at least $x_i$. Thus no left-side agent can improve on the direct route, and the mechanism's choice $a^*=x_r$ is optimal for this family, also covering $N_1=\emptyset$.

Suppose $x_r>t$. For each $i\in N_1$, the cost
\[
\min\{x_i,\ |x_i-a|+k(1-a)\}
\]
is continuous and piecewise linear in $a$. Its possible breakpoints occur at $a=x_i$ and where the pathway cost equals the direct-route cost. The latter conditions are
\[
\begin{aligned}
x_i-a+k(1-a)=x_i \quad &(a\le x_i),
    &&\text{giving }a=t,\\
a-x_i+k(1-a)=x_i \quad &(a\ge x_i),
    &&\text{giving }a=\frac{2x_i-k}{1-k}.
\end{aligned}
\]
Only points satisfying the indicated conditions and lying in $[0,x_r]$ can be actual breakpoints.

The common breakpoint $t$ need not be included in $S_1$. Indeed, for $a\in[0,t]$, every left-side agent takes the direct route at no greater cost, so social cost is constant. Immediately to the right of $t$, each agent with $x_i>t$ strictly benefits and has pathway cost decreasing with slope $-(1+k)$, while agents with $x_i\le t$ retain their direct-route costs. Since $x_r>t$, social cost strictly decreases there. Hence no point in $[0,t]$ is optimal.

A continuous piecewise linear function on $[0,x_r]$ has a minimum at an endpoint or a breakpoint. All remaining possible minimizers of this type belong to $S_1$: $x_r$ is an agent location, and all other relevant breakpoints are agent locations or the listed switching points. Therefore $S_1$ contains an optimal endpoint for the family $(a,1)$.

Reflecting the interval by $x\mapsto 1-x$ gives the symmetric argument for $(0,b)$. The common breakpoint is then $b=\frac{1}{1+k}$, and the other switching points are $b=\frac{2x_i-1}{1-k}$, precisely those used in $S_2$. Thus Method 2 also returns an optimal endpoint for its family. Comparing the two families proves the theorem.
\end{proof}

Note also that Mechanism~\ref{mec:opt-sc} is not strategyproof. For example, 
suppose there is one agent in $N_1$ and one in $N_2$, with 
$o=0.5$, $x_1=0.5-\epsilon$, and $x_2=0.5+\epsilon$, where $\epsilon>0$ is 
small and satisfies $k(0.5+\epsilon-\epsilon^2)+\epsilon^2 < 0.5-\epsilon$. 
The optimal solution is either $(x_1,1)$ or $(0,x_2)$. By symmetry, assume 
the mechanism outputs $(x_1,1)$. If agent $2$ misreports as 
$x_2'=x_2-\epsilon^2$, then the mechanism returns $(0,x_2')$, which reduces 
its cost from $0.5-\epsilon$ to $k(0.5+\epsilon-\epsilon^2)+\epsilon^2$. On the reported profile, the candidate $(x_1,1)$ has social cost exceeding that of $(0,x_2')$ by $(1+k)\epsilon^2$, so the change of output does not depend on tie-breaking.

\section{Deterministic Strategyproof Mechanisms}\label{sec:mec}

All mechanisms in this section are deterministic. We analyze SP/GSP mechanisms for maximum cost and social cost.
% in Sections \ref{subsec:max} and \ref{subsec:soc}, respectively. %For the upper bound, we provide group strategyproof mechanisms, and for the lower bound, we discuss strategyproof mechanisms. In this way, we show that these two types of mechanism share the same bounds.
We use the same notation as before:
$x_l, x_r, y_l, y_r, \delta_1, \delta_2$, etc.

\subsection{Maximum Cost}\label{subsec:max}

%In this section we discuss strategyproof and group strategyproof mechanisms for the maximum cost objective.

%The \textsc{TwoExtreme} Mechanism provides a linear upper bound for social cost.
We first consider the \textsc{TwoExtreme} mechanism as follows, which simply connects $x_r$ and $y_l$.

\begin{mechanism}[\textsc{TwoExtreme}]\label{mec:2-extreme}
Given a location profile $\mathbf x$, return $(a,b)=(x_r,y_l)$. 
\end{mechanism}

We prove that the \textsc{TwoExtreme} mechanism is group strategyproof and achieves a $\frac{3-k}{1+k}$-approximation  for maximum cost.

\begin{lemma}
     Mechanism \ref{mec:2-extreme} is group strategyproof.
\end{lemma}
\begin{proof}
Suppose a coalition $S$ can jointly misreport so that every member strictly benefits. By Corollary~\ref{coro}, all members belong to the same region; by symmetry, assume $S\subseteq N_1$. The right endpoint remains $y_l$. Let $a'$ be the new left endpoint.

If $a'\ge x_r$, then every true left-side location satisfies $x_i\le x_r\le a'$. Increasing the left endpoint from $x_r$ to $a'$ increases each such agent's pathway cost by $(1-k)(a'-x_r)$, so no member can strictly benefit. If $a'<x_r$, every agent at $x_r$ must belong to $S$. For any such agent, the pathway cost increases by $(1+k)(x_r-a')$, so its actual cost cannot decrease either. Both cases contradict the assumed strict improvement.
\end{proof}

\begin{theorem}
Mechanism \ref{mec:2-extreme} is group strategyproof and achieves a 
$\frac{3-k}{1+k}$-approximation for the maximum cost.
\end{theorem}

\begin{proof}
Group strategyproofness follows from the preceding lemma. By symmetry, assume $x_r+y_l\ge 1$, and write $ALG$ and $OPT$ for the maximum costs of the mechanism and an optimal solution. As shown in the proof of optimality of Mechanism~\ref{mec:opt-mc}, an optimal solution can be chosen with $b=1$. For any left-side agent at $v$, let
\[
T=1-y_l,\qquad
D(v)=\frac{x_r-v}{2}+k\left(1-\frac{x_r+v}{2}\right).
\]
The right-side agents and the two-agent bound established in that proof give
\begin{equation}\label{eq:mc-pair-lower}
OPT\ge T,\qquad OPT\ge\min\{v,D(v)\}.
\end{equation}
The first bound also holds when $N_2$ is empty, since then $T=0$.

Every right-side agent has cost at most $T$. Consider a left-side agent at $v$. If $v\le OPT$, its direct route already gives cost at most $OPT$. Otherwise, \eqref{eq:mc-pair-lower} implies $D(v)\le OPT$. Its pathway cost under TwoExtreme is
\[
P(v)=x_r-v+k(y_l-x_r)+1-y_l.
\]
Since $y_l\ge x_r$ and $0\le k<1$, we have
\begin{align*}
(1+k)P(v)
&=2D(v)+(1-k)T-k(1-k)(y_l-x_r)\\
&\le 2D(v)+(1-k)T
\le (3-k)OPT.
\end{align*}
Thus every agent's actual cost is at most $\frac{3-k}{1+k}OPT$, which proves the claimed bound, including when $OPT=0$.
\end{proof}

% \noindent \textbf{Remark.} The analysis for the approximation ratio of Mechanism \ref{mec:2-extreme} is tight.
% For instance, take $n=3$, $o=0.9$, and
% \[
% x_1 = 0.9-0.2\cdot\frac{1-k}{1+k},\quad x_2\to o^-,\quad x_3\to o^+.
% \]
% Then the optimal solution is $(0.9-0.1\cdot\frac{1-k}{1+k},\,1)$ with maximum cost $0.1\cdot \frac{1-k}{1+k}+k(0.1+0.1\cdot \frac{1-k}{1+k})=0.1$, whereas the mechanism returns $(0.9,0.9)$ with maximum cost $0.2\cdot\frac{1-k}{1+k}+0.1 = 0.1\cdot\frac{3-k}{1+k}$, giving the ratio $\frac{3-k}{1+k}$.

While the \textsc{TwoExtreme} mechanism is group strategyproof and achieves a $\frac{3-k}{1+k}$-approximation for the maximum cost (as proved above), we provide a novel mechanism that improves this ratio.
%Firstly we provide a group strategyproof mechanism and discuss the upper bound.

\begin{mechanism}[\textsc{CriticalExtreme}]\label{mec:critical-extreme}
Given a location profile $\mathbf x$, define
\(L(y)=\frac{2y-1}{1-k}\) and \(R(x)=\frac{1-2x}{1-k}\).
If $x_r+y_l \le 1$, return $(a,b)=\bigl(\max\{0,y_l-R(x_r)\},y_l\bigr)$;
otherwise return $(a,b)=\bigl(x_r,\min\{1,x_r+L(y_l)\}\bigr)$.
\end{mechanism}

% \begin{mechanism}\label{mec:2-extreme-improved}
% Given a location profile $\mathbf x$,  if $x_r+y_l\le 1$, return $(a, b) =(\max(0, 2x_r+y_l-1), y_l)$; else return $(a, b) = (x_r, \min(1, x_r+2y_l-1))$.
% \qquad if $2x_r + y_l\le 1$, return $(0, y_l)$;

% \qquad else return $(2x_r+y_l-1, y_l)$;

% \quad else:

% \qquad if $2y_l + x_r\ge 2$, return $(x_r, 1)$;

% \qquad else return $(x_r, x_r+2y_l-1)$.
% \end{mechanism}

The output is feasible in both branches. If $x_r+y_l\le1$, then $(1+k)x_r+(1-k)y_l\le x_r+y_l\le1$, implying $R(x_r)\ge y_l-x_r$ and hence $a\le x_r<o$. If $x_r+y_l>1$, then $(1-k)x_r+(1+k)y_l\ge x_r+y_l>1$, implying $L(y_l)\ge y_l-x_r$ and hence $b\ge y_l>o$. We next establish group strategyproofness and the approximation guarantee. 

\begin{lemma}
    Mechanism \ref{mec:critical-extreme} is group strategyproof.
\end{lemma}

\begin{proof}
By Corollary~\ref{coro}, a coalition whose members all strictly benefit must lie in one region. We give the argument for a coalition $S\subseteq N_1$; the other side is symmetric. The value $y_l$ remains fixed. Write $r=x_r$ for the true rightmost location in $N_1$ and $s$ for its value after the joint misreport, and put $T=1-y_l$.

First consider any report with $s+y_l\le 1$. The output has $b=y_l$ and $a=\max\{0,y_l-R(s)\}\le s$. Since $s<y_l$, we have
\[
(1+k)a+(1-k)y_l\le (1+k)s+(1-k)y_l\le s+y_l\le 1.
\]
For any true left-side location $v$, the pathway cost is therefore at least
\[
v-a+k(y_l-a)+1-y_l
=v+1-(1+k)a-(1-k)y_l\ge v.
\]
Hence such an output gives every left-side agent its direct-route cost.

For a report with $s+y_l>1$, the output is $(s,\min\{1,s+L(y_l)\})$. The cost from $s$ to the facility at $1$ via this pathway is
\[
h(s)=\max\{k(1-s),\ 2T-s\}.
\]
In particular, when $v\le s$, the pathway cost is
\[
s-v+h(s)=\max\{k+(1-k)s-v,\ 2T-v\},
\]
which is nondecreasing in $s$.

If $r+y_l\le 1$, every member initially has its direct-route cost. A new report with $s+y_l\le 1$ cannot improve it. If instead $s+y_l>1$, then $v\le r\le T<s$ for every member, and the displayed pathway cost is at least $2T-v\ge v$, so again no member benefits.

It remains to consider $r+y_l>1$. A new report with $s+y_l\le 1$ cannot improve any member's cost. Otherwise, if $s\ge r$, the displayed monotonicity applies to all true locations $v\le r$, so none can benefit. If $s<r$, every agent at $r$ must belong to $S$. Since $h$ is nonincreasing, its new pathway cost $r-s+h(s)$ is at least $h(r)$, and its actual cost cannot decrease. This rules out every profitable coalition in $N_1$. The symmetric argument applies to $N_2$; at $s+y_l=1$, either branch formula gives all agents their direct-route costs, so the tie-breaking rule does not affect the argument.
\end{proof}

\begin{theorem}\label{thm:333}
    Mechanism \ref{mec:critical-extreme} is group strategyproof and achieves a $\frac{2}{1+k}$-approximation for the maximum cost.
\end{theorem}

\begin{proof}
Group strategyproofness follows from the preceding lemma. If $x_r+y_l=1$, the maximum cost without a pathway is $x_r=1-y_l$. By Lemma~\ref{lem:basic}, no pathway can strictly reduce both extreme agents' costs, so the mechanism is optimal. This also covers an empty region, in which case all agents have zero direct-route cost.

Otherwise, by symmetry assume $x_r+y_l>1$. Use the same notation as in the preceding approximation analysis:
\[
T=1-y_l,\qquad
D(v)=\frac{x_r-v}{2}+k\left(1-\frac{x_r+v}{2}\right).
\]
The optimal maximum cost again satisfies \eqref{eq:mc-pair-lower}. Every right-side agent has cost at most $T\le OPT$. For a left-side agent at $v$, the mechanism's pathway cost is
\begin{align*}
P(v)
&=x_r-v+k(b-x_r)+1-b\\
&=\max\{x_r-v+k(1-x_r),\ 2T-v\},
\end{align*}
where $b=\min\{1,x_r+L(y_l)\}$.

If $v\le OPT$, the direct route gives the desired bound. Otherwise, $D(v)\le OPT$ by \eqref{eq:mc-pair-lower}, and
\begin{align*}
(1+k)\bigl(x_r-v+k(1-x_r)\bigr)
&=2D(v)-k(1-k)(1-x_r)\\
&\le 2D(v)\le 2OPT.
\end{align*}
The other term in $P(v)$ satisfies $2T-v\le T\le OPT$. Since $\frac{2}{1+k}\ge 1$, both terms are at most $\frac{2}{1+k}OPT$. Hence every agent's actual cost satisfies the claimed bound, including when $OPT=0$.
\end{proof}

Next, we prove a matching lower bound of $\frac{2}{1+k}$.

\begin{theorem}
    For any $0\le k<1$ and $n\ge 2$, no strategyproof mechanism can achieve an approximation ratio better than $\frac{2}{1+k}$ for the maximum cost. 
\end{theorem}

\begin{proof}
Fix $0\le k<1$ and $n\ge 2$. Suppose, for contradiction, that a strategyproof mechanism guarantees an approximation ratio $\rho<\frac{2}{1+k}$, where $\rho\ge 1$. Let
\[
\gamma=1-\frac{\rho(1+k)}{2}>0,\qquad d=\frac{1}{4},
\]
and choose a positive $\eta$ such that
\[
\eta<\min\left\{\frac{1}{4},\frac{\gamma^2d}{2(1+\rho k)}\right\}.
\]
Set $r=1-\eta$ and fix the obstacle at $o=1-\frac{\eta}{2}$.

Consider a profile with agent $1$ at $r-t$ and all other agents at $r$, where $\frac{\gamma d}{2}\le t\le d$. The feasible pathway $\left(r-\frac{t}{2},1\right)$ gives every agent cost at most $\frac{(1+k)t}{2}+k\eta$. Thus the mechanism's maximum cost is at most
\[
B(t)=\rho\left(\frac{(1+k)t}{2}+k\eta\right)
=(1-\gamma)t+\rho k\eta<t-\eta,
\]
where the strict inequality follows from $\gamma t\ge\frac{\gamma^2d}{2}>(1+\rho k)\eta$. Since every direct-route cost is at least $r-d>\frac{1}{2}>t$, all agents strictly prefer the pathway at the mechanism's output. Moreover, its left endpoint must lie strictly between $r-t$ and $r$: otherwise, one of the two occupied locations has pathway cost at least $t$, contradicting the bound $B(t)<t$.

Now take $t=d$, put $x=r-d$, and denote the output by $(a,b)$. We have $x<a<r$. Since the pathway cost of agent $1$ is at least $a-x$, the distance $d'=r-a$ satisfies
\[
d'\ge d-B(d)=\gamma d-\rho k\eta>\frac{\gamma d}{2},
\]
using $\rho k\eta<\frac{\gamma^2d}{2}$ and $0<\gamma<1$. Also $d'<d$.

Move agent $1$ to $a$, leaving every other location and the obstacle unchanged, and let $(a',b')$ be the new output. Applying the preceding bounds with $t=d'$ gives $a<a'<r$ and maximum cost at most $B(d')<d'-\eta$.

To apply strategyproofness in both directions, write
\[
H=1+(1-k)(a-b),\qquad H'=1+(1-k)(a'-b').
\]
Agent $1$ uses the pathway at each truthful profile. Since $x<a<a'$, its costs and the two strategyproofness inequalities give
\[
\begin{aligned}
H-x&\le\min\{x,H'-x\}\le H'-x,\\
H'-a&\le\min\{a,H-a\}\le H-a.
\end{aligned}
\]
Hence $H=H'$. As $k<1$, this implies
\[
a'-a=b'-b<1-o=\frac{\eta}{2}.
\]
At the new profile, an agent at $r$ uses the pathway and therefore has cost at least
\[
r-a'=d'-(a'-a)>d'-\frac{\eta}{2}>B(d').
\]
This contradicts the approximation guarantee. The claimed lower bound follows.
\end{proof}

%Therefore for maximum cost, we have a lower bound of $\frac{2}{1+k}$.

\subsection{Social Cost}\label{subsec:soc}

We next discuss SP/GSP mechanisms for the social cost objective. 
Recall that the \textsc{TwoExtreme} mechanism (Mechanism \ref{mec:2-extreme}), which simply connects $(x_r,y_l)$, is group strategyproof. We prove an approximation ratio of $\frac{n}{1+k(n-1)}$ for the social cost. This ratio is a constant when $k=\Omega(1)$ and is linear when $k=0$.

%Although \cite{chan2025mechanism} established strategyproofness for their cost model, we reprove it because the definition of an agent's cost differs in our model.\redcomment{maybe move GSP proof to appendix?}

\begin{theorem}\label{thm:sccc}
    Mechanism \ref{mec:2-extreme} is group strategyproof and achieves a $\frac{n}{1+k(n-1)}$-approximation for the social cost.
\end{theorem}

\begin{proof}
Group strategyproofness has already been proved. By Lemma~\ref{lem:opt-basic} and symmetry, choose an optimal solution $(a^*,1)$ with $a^*\le x_r$. Let $U\subseteq N_1$ consist of the agents who strictly benefit from this optimal pathway, and put $m=|U|$. Every agent outside $U$ has its direct-route cost under the optimal solution, so its cost cannot increase under TwoExtreme. If $m=0$, the mechanism is therefore optimal.

Suppose $m\ge 1$, and write $ALG$ and $OPT$ for the social costs of the mechanism and the optimal solution. Writing $h=k(1-a^*)$, strict improvement requires $h<a^*$, or equivalently $a^*>\frac{k}{1+k}$. Indeed, $|x_i-a^*|+h<x_i$ implies $h<a^*$ by the triangle inequality. Since $x_r\ge a^*$, its pathway cost $x_r-a^*+h$ is strictly less than $x_r$, so an agent at $x_r$ belongs to $U$. Define
\[
\delta=x_r-a^*,\qquad T=1-y_l.
\]
For any $i\in U$, the mechanism's actual cost is at most its pathway cost, and hence
\begin{align*}
&cost(x_r,y_l,x_i)-cost(a^*,1,x_i)\\
&\quad\le x_r-x_i+k(y_l-x_r)+1-y_l
             -|x_i-a^*|-k(1-a^*)\\
&\quad=(1-k)(\delta+T)-(x_i-a^*)-|x_i-a^*|\\
&\quad\le(1-k)(\delta+T).
\end{align*}

\textbf{Case 1: $N_2\ne\emptyset$.}
Here $m\le n-1$. Under the optimal solution, an agent at $y_l$ has cost $T$, while the total cost of the agents in $U$ is at least $\delta+mk(1-a^*)$: their total distance to $a^*$ is at least the distance $\delta$ of the agent at $x_r$. Since $1-a^*\ge\delta+T$, we obtain
\[
OPT\ge\delta+mk(1-a^*)+T\ge(1+mk)(\delta+T).
\]
Summing the preceding cost comparisons gives
\[
ALG\le OPT+m(1-k)(\delta+T)
\le\frac{m+1}{1+mk}OPT.
\]

\textbf{Case 2: $N_2=\emptyset$.}
Then $y_l=1$ and $T=0$. The agent at $x_r$ has mechanism cost at most $k(1-x_r)$, which is no larger than its optimal-solution cost $\delta+k(1-a^*)$. Thus only the other $m-1$ agents in $U$ can contribute a cost increase, yielding
\[
ALG\le OPT+(m-1)(1-k)\delta.
\]
Moreover,
\[
OPT\ge\delta+mk(1-a^*)
\ge\bigl(1+k(m-1)\bigr)\delta.
\]
Therefore,
\[
ALG\le\frac{m}{1+k(m-1)}OPT.
\]

The function $\frac{q}{1+k(q-1)}$ is nondecreasing for $q\ge1$. Since $m+1\le n$ in Case 1 and $m\le n$ in Case 2, both bounds are at most $\frac{n}{1+k(n-1)}OPT$. The proof also covers $OPT=0$, as no division by $OPT$ was used.
\end{proof}

% \noindent\emph{Proof sketch.}
% \ma{pending}
% \qed

% \begin{proof}[Proof Sketch.]
%      With loss of generality, we just suppose the optimal solution is $(a^*, 1)$. We can divide agents from $N_1$ into two groups. $G_1$ contains agents in $N_1$ that do not use the pathway $(a^*, 1)$; $G_2$ contains agents in $N_1$ that use the pathway $(a^*, 1)$. Also we can divide agents from $N_1$ into another two groups. $G_1'$ contains agents in $N_1$ that do not use the pathway $(x_r, y_l)$; $G_2'$ contains agents in $N_1$ that use the pathway $(x_r, y_l)$. It is easy to see that $G_1\subseteq G_1'$ and $G_2'\subseteq G_2$. We use $x_u$ to denote the smallest coordinate of $G_2$ and $x_v$ to denote the largest coordinate of $G_1$. Obviously we know $x_u\le a^*\le x_r$. We use $x_u'$ to denote the smallest coordinate of $G_2'$ and $x_v'$ to denote the largest coordinate of $G_1'$. Obviously $x_v'\ge x_v$ and $x_u'\ge x_u$.
%      Then we discuss the cases when $y_l=1$ and $y_l=o$ to prove the upper bounds of $\frac{\text{ALG}}{\text{OPT}}$ as $\frac{n-1}{nk+1}$ and $\frac{n}{k(n-1)+1}$ respectively.. 
% \end{proof}

The approximation bound is tight. For $n\ge2$, let $o=\frac{1}{2}$, place $n-1$ agents at $o-\epsilon$ and one at $o+\epsilon$, and choose $0<\epsilon<\frac{1-k}{2(1+k)}$. TwoExtreme returns $(o-\epsilon,o+\epsilon)$, and every agent has cost $\frac{1}{2}-\epsilon$. By Lemma~\ref{lem:opt-basic}, an optimal solution is $(o-\epsilon,1)$: it minimizes each left-side agent's cost, and serving the single right-side agent instead cannot give a smaller social cost. Thus
\[
ALG=n\left(\frac{1}{2}-\epsilon\right),\qquad
OPT=(n-1)k\left(\frac{1}{2}+\epsilon\right)+\frac{1}{2}-\epsilon.
\]
As $\epsilon\to0^+$, the ratio tends to $\frac{n}{1+k(n-1)}$. For $n=1$, the mechanism is optimal and the bound equals $1$.

We next establish lower bounds for social cost.

\begin{theorem}\label{thm:sc-lb}
Let $n\ge2$, and for $0<k<1$ define
\[
\beta(k)=\frac{-1-3k+\sqrt{16k^3+33k^2+14k+1}}{2k(k+1)}.
\]
Every strategyproof mechanism has a social-cost approximation ratio
\begin{itemize}
    \item at least $\beta(k)$ when $0<k<1$;
    \item at least $\Lambda_m(k)$ for every integer $3\le m\le n$ with $0<k<\frac{m-2}{m}$, where
    \[
    A_m=m-1-mk,\qquad
    \Lambda_m(k)=\frac{2A_m(1+mk)}{1+k+\sqrt{(1+k)^2+4k(m-1)A_m^2(1+mk)}};
    \]
    \item at least $\max\{2,n-1\}$ when $k=0$.
\end{itemize}
\end{theorem}

The first two bounds are complementary: we may take the maximum over all applicable bounds and all admissible values of $m$. We use the following continuity property, which follows from strategyproofness rather than from any regularity assumption on the mechanism.

\begin{lemma}\label{lem:truthful-cost-continuity}
Fix a nonempty set $S$ of agents and the locations of all other agents. Suppose the agents in $S$ share a common location $t$ within a fixed region. Under any strategyproof mechanism, their common truthful cost $c_S(t)$ satisfies
\[
|c_S(t)-c_S(s)|\le |t-s|
\]
for any two locations $s,t$ in that region.
\end{lemma}

\begin{proof}
For a fixed pathway, the cost of an agent is $1$-Lipschitz in its location within either region: both route costs are $1$-Lipschitz, as is their minimum. By Lemma~\ref{lemma:partial-gsp}, the agents in $S$ cannot strictly decrease their common cost by jointly reporting $s$ instead of $t$. Evaluating the output for reports $s$ at the true location $t$ therefore gives
\[
c_S(t)\le c_S(s)+|t-s|.
\]
Interchanging $s$ and $t$ proves the claim.
\end{proof}

\begin{lemma}\label{lem:sc-lb-2}
Fix $0\le k<1$ and $\frac{k}{1+k}<x<\frac{1}{2}$. For any $n\ge2$, every strategyproof mechanism has a social-cost approximation ratio at least
\[
\min\left\{\frac{1+2k(1-x)}{2x+k},\frac{2x}{x+k(1-x)}\right\}.
\]
\end{lemma}

\begin{proof}
Suppose a strategyproof mechanism guarantees a ratio $\rho$ strictly below both expressions. Set $o=\frac{1}{2}$, place two agents at $x$ and $1-x$, and place any remaining agents at $0$, where their costs are always zero. The optimal social cost is $x+k(1-x)$, whereas an output from which neither active agent strictly benefits has social cost $2x$. Hence at least one active agent must strictly benefit. By Lemma~\ref{lem:basic}, they cannot both do so. By symmetry between the two active agents, suppose the agent at $x$ has cost $x$.

Choose a finite $z\in(x,\frac{1}{2})$ sufficiently close to $\frac{1}{2}$ that
\[
\frac{z+k(1-x)}{x+k(1-z)}>\rho.
\]
Keep all other locations fixed, and let the left active agent's true location vary over $t\in[x,z]$. Write $c(t)$ for its truthful cost and $u(t)=t-c(t)$ for its improvement over the direct route. Lemma~\ref{lem:truthful-cost-continuity} implies that $u$ is continuous, and $u(x)=0$.

Whenever $u(t)>0$, the right active agent cannot strictly benefit from the output, so the social cost is $t+x-u(t)$. The feasible pathway $(t,1)$ has social cost $k(1-t)+x$. Thus the approximation guarantee requires
\[
u(t)\ge t+x-\rho\bigl(k(1-t)+x\bigr)
\ge 2x-\rho\bigl(k(1-x)+x\bigr)>0.
\]
The final quantity is a fixed positive constant. A continuous function starting at zero cannot become positive while avoiding every value between zero and that constant. Therefore $u(t)=0$ throughout $[x,z]$.

At $t=z$, the left active agent consequently has cost $z$. The right active agent has cost at least $k(1-x)$, its minimum possible cost under any pathway. The mechanism's social cost is therefore at least $z+k(1-x)$, while the optimum is at most $x+k(1-z)$. This contradicts the choice of $z$.
\end{proof}

\begin{lemma}\label{lem:sc-lb-1}
When $k=0$, every strategyproof mechanism has a social-cost approximation ratio at least $n-1$.
\end{lemma}

\begin{proof}
The claim is immediate for $n\le2$. For $n\ge3$, put $p=n-1$ and suppose a strategyproof mechanism guarantees a ratio $1\le\rho<p$. Choose
\[
r=\frac{1}{2},\qquad
\epsilon=\frac{r}{4\rho p},\qquad
\alpha=\frac{\epsilon}{2},\qquad
0<\eta<\min\left\{\frac{1}{4},\frac{\alpha(p-\rho)}{2p}\right\}.
\]
Fix the obstacle at $o=\eta$ and put all agents in $N_2$. For convenience, measure their locations by distance $t=1-x$ from the facility at $1$. If the mechanism returns $(a,b)$, write $z=1-b$. In these coordinates, $0\le a<\eta$, $0\le z<1-\eta$, and an agent at $t$ has cost
\[
\min\{t,|t-z|+a\}.
\]

Place $p$ agents at a common distance $s\in[\epsilon,r]$ and the remaining agent at $r$. Let $g(s)$ be the common truthful cost of the group. This function is continuous by Lemma~\ref{lem:truthful-cost-continuity}. At $s=r$, a zero-cost solution exists, so the approximation guarantee forces $g(r)=0$.

At $s=\epsilon$, choosing $a=0$ and $z=r$ gives social cost $p\epsilon$. Thus the mechanism's social cost is at most $\rho p\epsilon=\frac{r}{4}$. The agent at $r$ must use the pathway, and its cost bound implies $z\ge\frac{3r}{4}$. Since $\epsilon\le\frac{r}{4}$, the group prefers the direct route, giving $g(\epsilon)=\epsilon$.

By continuity, there is some $s\in(\epsilon,r)$ with $g(s)=\alpha$. Denote the output at this profile by $(a,z)$ in the new coordinates and put $D=|s-z|$. Since $\alpha<s$, the group uses the pathway, so
\[
D=\alpha-a>\alpha-\eta>\frac{\alpha}{2},\qquad
z\ge s-D>\epsilon-\alpha=\alpha>\eta.
\]
Now move only the remaining agent from $r$ to $z$, and write $(a',z')$ for the new output. If this agent, when truly at $z$, reports $r$, its cost is $a$. Strategyproofness therefore implies that its new cost is at most $a<\eta$. Because $z>\eta$, it must use the new pathway, and consequently $|z-z'|\le a<\eta$.

Each group member's new cost is at least $D-\eta$, since $D<\alpha<s$ and $|s-z'|\ge D-|z-z'|$. On the other hand, choosing fee $0$ and center $s$ gives a feasible pathway with social cost at most $D$. The mechanism's new social cost is thus at least
\[
p(D-\eta)>\rho D,
\]
where the strict inequality follows from $D>\frac{\alpha}{2}$ and $p\eta<\frac{\alpha(p-\rho)}{2}$. This contradicts the approximation guarantee.
\end{proof}

\begin{lemma}\label{lem:sc-positive-k}
Let $n\ge3$ and $0<k<\frac{n-2}{n}$. Every strategyproof mechanism has a social-cost approximation ratio at least $\Lambda_n(k)$.
\end{lemma}

\begin{proof}
Put $p=n-1$ and $A=p-nk>1$. The number $\Lambda=\Lambda_n(k)$ is the positive root of
\[
 kpA\Lambda^2+(1+k)\Lambda-A(1+nk)=0.
\]
The polynomial takes values $(1+k)(1-A)<0$ at $1$ and $kpA(A^2-1)>0$ at $A$, so $1<\Lambda<A$. Suppose a strategyproof mechanism has approximation ratio $1\le\rho<\Lambda$. Define
\[
 R(\gamma)=\frac{k+A\gamma}{\gamma+kA},\qquad
 \gamma_0=\frac{1+nk}{\rho p}<1.
\]
Here $\gamma_0<1$ follows from $1+nk<p$ and $\rho\ge1$. The polynomial inequality at $\rho$ gives $R(\gamma_0)>\rho$. Moreover, $\gamma_0>\rho k$: otherwise, since $A>\rho$,
\[
 (A-\rho)\gamma_0+k(1-\rho A)
 \le k(1-\rho^2)\le0,
\]
contrary to $R(\gamma_0)>\rho$. Choose $\gamma<\gamma_0$ sufficiently close to $\gamma_0$ that $\gamma>\rho k$ and $R(\gamma)>\rho$, and then choose $\epsilon$ with
\[
 \gamma<\frac{\epsilon}{1-\epsilon}<\gamma_0.
\]
In particular, $\epsilon<\frac12$. Choose $r<1$ sufficiently close to $1$ that $r>2\epsilon$ and
\begin{equation}\label{eq:sc-positive-start}
 r-\epsilon+nk(1-\epsilon)
 >\rho\bigl(p\epsilon+k(1-r)\bigr).
\end{equation}
Such an $r$ exists because the limiting inequality at $r=1$ is equivalent to $\frac{\epsilon}{1-\epsilon}<\gamma_0$. Finally, put
\[
 M=\frac{(A-\rho)\gamma+k(1-\rho A)}{1-k}>0,
 \qquad
 0<\eta<\min\left\{1-r,\frac{M(1-r)}{4p}\right\}.
\]
Fix the obstacle at $o=\eta$, put all agents in $N_2$, and use distances $t=1-x$ from the facility at $1$. Write the output as $(a,z)$ in these coordinates, where $z=1-b$. The cost at $t$ is
\[
 \min\{t,|t-z|+h\},\qquad
 h=k(1-z)+(1-k)a,\qquad 0\le a<\eta.
\]

Place $p$ agents at $s\in[\epsilon,r]$ and one at $r$, and denote the group's common truthful cost by $g(s)$. At $s=\epsilon$, the feasible output $a=0,z=r$ gives $OPT\le p\epsilon+k(1-r)$. If the group strictly benefited, then $|\epsilon-z|+h<\epsilon$, which implies $z<2\epsilon<r$ and $h<z$. The agent at $r$ would also use the pathway. For $z\le\epsilon$, the social cost would then satisfy
\[
 ALG=p(\epsilon-z)+r-z+nh
 \ge r+p\epsilon+nk-n(1+k)z
 \ge r-\epsilon+nk(1-\epsilon).
\]
For $z\ge\epsilon$, the same lower bound follows from
\[
 ALG=p(z-\epsilon)+r-z+nh
 \ge r-p\epsilon+nk+(p-1-nk)z,
\]
since $p-1-nk=A-1>0$. Both cases contradict \eqref{eq:sc-positive-start}. Hence $g(\epsilon)=\epsilon$.

At $s=r$, placing the pathway at $a=0,z=r$ gives $g(r)\le\rho k(1-r)$. By Lemma~\ref{lem:truthful-cost-continuity}, the function $g(s)-\gamma(1-s)$ is continuous, positive at $\epsilon$, and negative at $r$. Thus for some $s\in(\epsilon,r)$,
\[
 g(s)=\gamma(1-s)<s.
\]
The strict inequality follows from $\gamma<\frac{\epsilon}{1-\epsilon}<\frac{s}{1-s}$. Fix its output $(a,z)$ and write $D=|s-z|$ and $w=1-s$. The group uses the pathway, so $g(s)=D+h$, and
\begin{equation}\label{eq:sc-positive-distance}
 D\le\frac{(\gamma-k)w}{1-k},
\end{equation}
since $h\ge k(1-z)\ge k(w-D)$. Also $h<z$, as follows from $D+h<s$.

Now move only the remaining agent from $r$ to $z$, and write the new output as $(a',z')$, with $h'=k(1-z')+(1-k)a'$. By reporting $r$, this agent could obtain cost $h<z$. Strategyproofness therefore implies
\[
 |z-z'|+h'\le h,
 \qquad
 |z-z'|-k(z'-z)\le(1-k)(a-a').
\]
The left-hand side is at least $(1-k)|z-z'|$, so $a'\le a$ and $|z-z'|\le a-a'<\eta$. Moreover, $|h'-h|\le k|z'-z|+(1-k)|a'-a|\le a-a'<\eta$. The triangle inequality therefore bounds the change in each group member's pathway cost by $2\eta$. Taking the minimum with its unchanged direct-route cost preserves this bound, so its new actual cost is at least $g(s)-2\eta$. The moved agent's cost is at least $k(1-z)$: this is its minimum pathway cost, and $k(1-z)\le h<z$.

The new social cost is therefore at least $p\gamma w+k(w-D)-2p\eta$. The feasible output $a=0,z=s$ gives a new optimal cost of at most $nkw+D$. Using \eqref{eq:sc-positive-distance} and the identity $p\gamma+k-\rho nk-\frac{(\rho+k)(\gamma-k)}{1-k}=M$, we obtain
\begin{align*}
 ALG-\rho OPT
 &\ge (p\gamma+k-\rho nk)w-(\rho+k)D-2p\eta\\
 &\ge Mw-2p\eta
 \ge M(1-r)-2p\eta>0.
\end{align*}
This contradicts the approximation guarantee.
\end{proof}

\begin{proof}[Proof of Theorem~\ref{thm:sc-lb}]
For $0<k<1$, the first expression in Lemma~\ref{lem:sc-lb-2} is strictly decreasing in $x$, while the second is strictly increasing. At $x=\frac{k}{1+k}$, the second equals $1$ and the first is larger; at $x=\frac{1}{2}$, the first equals $1$ and the second is larger. Their unique intersection therefore lies strictly inside the lemma's admissible interval and maximizes their minimum.

Let $E$ be their common value. Eliminating $x$ from the two expressions gives
\[
k(k+1)E^2+(1+3k)E-2(2k+1)=0.
\]
Its positive root is $E=\beta(k)>1$, proving the first bound. Lemma~\ref{lem:sc-positive-k} gives the second bound when $m=n$. For any $3\le m<n$, fix the remaining $n-m$ agents at the facility $0$. Their costs are always zero, and the induced mechanism on the $m$ active agents remains strategyproof with the same approximation guarantee. Applying the lemma to these agents gives $\Lambda_m(k)$ as well.

Finally, when $k=0$, Lemma~\ref{lem:sc-lb-2} with $x=\frac{1}{4}$ gives a lower bound of $2$, and Lemma~\ref{lem:sc-lb-1} gives $n-1$. Combining them completes the proof.
\end{proof}

The new bound satisfies $\lim_{k\to0^+}\Lambda_m(k)=m-1$, so it strengthens $\beta(k)$ for sufficiently small positive $k$ whenever $m\ge4$. Taking the maximum over $m\le n$ also retains bounds obtained from smaller active populations. For example, when $n=6$ and $k=0.01$, $\Lambda_6(k)\approx2.993$, whereas $\beta(k)\approx1.944$. At $k=0$, the lower bound $\max\{2,n-1\}$ matches the upper bound for $n=2$ and is asymptotically tight as $n$ grows.

\section{Randomized Strategyproof Mechanisms}\label{sec:randomized}

We now allow randomized pathways and require strategyproofness in expectation. Randomization gives a social-cost guarantee independent of the number of agents, even when $k=0$. For maximum cost, we establish a randomized lower bound and retain the deterministic upper bound as a benchmark.

\subsection{Maximum Cost}

\begin{lemma}\label{lem:random-mc-pointwise}
Let $c(v)=\min\{v,|v-a|+h\}$, where $h\ge0$. For $0<2d<t$,
\[
 \max\{c(t-d),c(t+d)\}
 \ge d+\left(1-\frac{2d}{t}\right)c(t).
\]
\end{lemma}
\begin{proof}
If $c(t-d)=t-d$, the right-hand side is at most $t-d$, since $c(t)\le t$. Otherwise, if $a\ge t$, the pathway cost at $t-d$ is $d+|t-a|+h\ge d+c(t)$. If $a<t$, then $c(t+d)=d+c(t)$. These cases prove the claim.
\end{proof}

\begin{theorem}\label{thm:random-mc-lower}
For $n\ge2$ and $0\le k<1$, every randomized strategyproof mechanism has a maximum-cost approximation ratio at least \(\frac{3+2k}{2+3k}\).
Together with Theorem~\ref{thm:333}, this gives lower and upper bounds of $\frac{3+2k}{2+3k}$ and $\frac{2}{1+k}$, respectively. At $k=0$, these bounds are $\frac32$ and $2$.
\end{theorem}
\begin{proof}
Suppose a mechanism guarantees a finite ratio $\rho$. Choose $0<d<\frac18$, put $\eta=d^2$ and $o=1-\frac{\eta}{2}$, and place two agents at
\[
 L=1-2d-\eta,\qquad R=1-d-\eta.
\]
Place all remaining agents at $0$. The locations $L-d,L,R,R+d$ all lie in the left region. At any output $(a,b)$, a left-side agent at $v$ has cost
\[
 c(v)=\min\{v,|v-a|+h\},\qquad
 h=k(b-a)+1-b\ge k(1-a).
\]
On the initial profile, the sum of the two active agents' costs is at least $d+2k(d+\eta)$. Indeed, if either agent takes the direct route, its cost alone is at least $L=1-2d-d^2>d+2k(d+d^2)$ for $0<d<\frac18$. Otherwise, their total cost is at least
\[
 |L-a|+|R-a|+2k(1-a).
\]
This expression is nonincreasing up to $a=R$ and increasing thereafter: its slopes on the intervals separated by $L$ and $R$ are $-2-2k$, $-2k$, and $2-2k$. Thus its minimum over $a\in[0,1]$ is attained at $a=R$ and equals $d+2k(d+\eta)$. Writing $e_L,e_R$ for their initial expected costs, we obtain
\begin{equation}\label{eq:random-mc-base}
 e_L+e_R\ge d+2k(d+\eta).
\end{equation}

Next, change only the agent at $L$ to $L-d$. Strategyproofness implies that the expected cost evaluated at the old location $L$, under the new output distribution, is at least $e_L$. Applying Lemma~\ref{lem:random-mc-pointwise} to each realized output gives an expected maximum cost of at least
\[
 d+\left(1-\frac{2d}{L}\right)e_L.
\]
The feasible pathway $(L,1)$ has maximum cost at most $(1+2k)d+k\eta$ on this new profile. Similarly, moving only the agent at $R$ to $R+d$ and comparing with the pathway $(R,1)$ yields
\begin{align*}
 d+\left(1-\frac{2d}{L}\right)e_L
 &\le\rho\bigl((1+2k)d+k\eta\bigr),\\
 d+\left(1-\frac{2d}{R}\right)e_R
 &\le\rho\bigl((1+k)d+k\eta\bigr).
\end{align*}
Rearranging these inequalities and using \eqref{eq:random-mc-base}, we get
\[
 d+2k(d+\eta)
 \le
 \frac{\rho((1+2k)d+k\eta)-d}{1-\frac{2d}{L}}
 +\frac{\rho((1+k)d+k\eta)-d}{1-\frac{2d}{R}}.
\]
Divide by $d$ and let $d\to0$. Since $\eta=d^2$ and $L,R\to1$, it follows that
$1+2k\le\rho(2+3k)-2$, proving the theorem. Each comparison uses a fixed positive $d$ and the same obstacle; the limit only selects increasingly difficult feasible instances.
\end{proof}

\subsection{Social Cost: A Power-Proportional Mechanism}

Let $t_i=x_i$ for $i\in N_1$ and $t_i=1-x_i$ for $i\in N_2$. The smallest cost agent $i$ can attain under any feasible pathway is
\[
 m_i=\min\{t_i,k(1-t_i)\},
\]
attained by $(x_i,1)$ in the left region and by $(0,x_i)$ in the right region. Define the maximum possible improvement
\[
 q_i=t_i-m_i=\max\{(1+k)t_i-k,0\}.
\]
These expressions follow by minimizing $|t_i-s|+k(1-s)$ over the endpoint distance $s$; since $k<1$, its minimum is attained at $s=t_i$.

\begin{mechanism}[\textsc{PowerProportional}]\label{mec:power-proportional}
Set $\theta=\frac{1-k}{1+k}$ and $w_i=q_i^\theta$. If $W=\sum_iw_i>0$, select agent $i$ with probability $\frac{w_i}{W}$ and output $(x_i,1)$ if $i\in N_1$, or $(0,x_i)$ if $i\in N_2$. If $W=0$, output $(0,1)$.
\end{mechanism}

At $k=0$, this selects an agent with probability proportional to $t_i$. It is the counterpart of the proportional mechanism of \cite{Chan023} under the change of coordinates discussed in Section~\ref{sec:improve}. The exponent in Mechanism~\ref{mec:power-proportional} accounts for the positive pathway cost.

\begin{theorem}\label{thm:power-sp}
Mechanism~\ref{mec:power-proportional} is strategyproof in expectation for every $0\le k<1$.
\end{theorem}
\begin{proof}
Fix an agent's true endpoint distance $t$ and the other reports. If $q=\max\{(1+k)t-k,0\}=0$, every feasible pathway gives this agent cost $t$, so no deviation helps. Suppose $q>0$, and put $m=k(1-t)=t-q$ and $w=q^\theta$. For each other agent $j$, let $c_j$ be the true cost incurred when $j$ is selected, and set
\[
 S=\sum_{j\ne i}w_j,\qquad
 A=\sum_{j\ne i}w_j(c_j-m).
\]
The selected agent's pathway depends only on its own report. Thus the other agents' weights and their selected pathways remain fixed under the deviation. Their selection probabilities can change, but only through the common normalizing denominator. This lets us collect their contributions in the constants $S$ and $A$. Since $m\le c_j\le t$, we have $0\le A\le Sq$, and the truthful expected cost is
\[
 m+\frac{A}{S+w}.
\]
A same-region report at endpoint distance $r$ changes this agent's weight to $w'$ and its cost when selected to $m+e$, where $e\ge0$. If all weights after the deviation vanish, the fallback gives cost $t$, which cannot improve on truthfulness. Otherwise, the expected cost after the deviation is
\[
 m+\frac{A+w'e}{S+w'}.
\]
Let $\Delta_i$ denote the deviating expected cost minus the truthful expected cost. Subtracting the two expressions gives
\[
 \Delta_i=\frac{w'e(S+w)-(w'-w)A}{(S+w')(S+w)}.
\]
For $r\le t$, we have $w'\le w$, so both contributions to the numerator are nonnegative. For $r>t$, put $q'=(1+k)r-k>q$. When the deviating agent is selected, its pathway cost before taking the minimum with the direct route is
\[
 |t-r|+k(1-r)=m+(1-k)(r-t).
\]
Since the direct-route cost is $t=m+q$, its additional cost when selected is
\[
 e=\min\{q,(1-k)(r-t)\}
   =\min\{q,\theta(q'-q)\}.
\]
The elementary inequality $(1+u)^{-\theta}\ge1-\theta u$ for $u\ge0$ implies
\[
 1-\left(\frac{q}{q'}\right)^\theta
 \le\theta\frac{q'-q}{q}.
\]
The quantity $q\left(1-\frac{w}{w'}\right)$ is at most both $q$ and $\theta(q'-q)$. It is therefore at most their minimum $e$, giving
\[
 w'e\ge(w'-w)q\ge\frac{(w'-w)A}{S+w}.
\]
This makes the numerator of $\Delta_i$ nonnegative. Thus the extra cost when the agent is selected compensates for any benefit from increasing its selection probability.
\end{proof}

\begin{theorem}\label{thm:power-sc}
For $0<k<1$, Mechanism~\ref{mec:power-proportional} has social-cost approximation ratio at most $1+C(k)$, where
\begin{equation}\label{eq:random-sc-C}
 C(k)=(1+k)\left(\frac{4}{1-k}\right)^{\frac{1-k}{1+k}}\le4.
\end{equation}
In particular, it is a $5$-approximation independent of $n$ and $k$.
\end{theorem}
\begin{proof}
Write $B=1+k$ and $\theta=\frac{1-k}{1+k}$. Fix an optimal pathway of the form $(a,1)$ or $(0,b)$, which exists by Lemma~\ref{lem:opt-basic}. Let $c_i^*$ be agent $i$'s optimal cost and $c_{ij}$ its cost when the mechanism selects $j$. We prove
\begin{equation}\label{eq:power-pair-bound}
 w_j(c_{ij}-c_i^*)\le C(k)w_i c_j^*
 \qquad\text{for all }i,j.
\end{equation}
The left-hand side measures the weighted excess cost of agent $i$ when $j$ is selected. The right-hand side bounds it using $j$'s optimal cost. Summing over both indices will give the total approximation guarantee.

If $c_i^*=t_i$, the left-hand side is nonpositive. If $w_j=0$, it is zero. Otherwise, $c_i^*<t_i$ implies $q_i>0$, since $q_i$ is the maximum possible improvement, and $w_j>0$ implies $q_j>0$. All divisions by these quantities below are therefore valid.

First suppose $c_j^*=t_j$. Since $c_{ij}-c_i^*\le q_i$, weighted arithmetic--geometric mean gives
\begin{align*}
 \frac{w_jq_i}{w_i}
 &=q_i^{1-\theta}q_j^\theta
 \le(1-\theta)q_i+\theta q_j\\
 &\le\frac{2k+(1-k)q_j}{1+k}
 \le 2\frac{q_j+k}{1+k}=2t_j.
\end{align*}
Here $q_i\le1$. Thus coefficient $2$ suffices in this case.

Now suppose both agents strictly use the optimal pathway. They lie in the same region. In endpoint-distance coordinates, let its center be $s$, so
\[
 c_i^*=|t_i-s|+k(1-s),\qquad
 c_j^*=|t_j-s|+k(1-s).
\]
Selecting $j$ gives $c_{ij}=\min\{t_i,|t_i-t_j|+k(1-t_j)\}$. The triangle inequality and the identity $k(1-t_j)=k(1-s)+k(s-t_j)$ yield
\begin{align*}
 c_{ij}&\le c_i^*+|s-t_j|+k(s-t_j)
          \le c_i^*+B c_j^*,\\
 c_{ij}-c_i^*&\le\min\{q_i,Bc_j^*\}.
\end{align*}
Also, $t_j\le t_i+c_i^*+c_j^*$ and $s\le t_i+c_i^*<2t_i$. Recalling $Bt_i=q_i+k$ and $Bt_j=q_j+k$, the first inequality gives
\[
 q_j=Bt_j-k\le2Bt_i+B c_j^*-k=2q_i+k+B c_j^*.
\]
The second, together with $c_j^*\ge k(1-s)$, gives
\[
 Bc_j^*\ge Bk(1-2t_i)=k(1-k)-2kq_i,
 \qquad
 k\le\frac{Bc_j^*+2kq_i}{1-k}.
\]
Substituting this upper bound for the additive $k$ in the preceding bound on $q_j$ yields
\begin{align*}
 q_j
 &\le2q_i+B c_j^*+\frac{Bc_j^*+2kq_i}{1-k}\\
 &=\frac{2}{1-k}q_i+\frac{B(2-k)}{1-k}c_j^*
 \le\frac{2}{1-k}q_i+\frac{2B}{1-k}c_j^*.
\end{align*}
For $k>0$, the center satisfies $s<1$, so $c_j^*>0$. Set $v=\frac{c_j^*}{q_i}>0$. It follows that
\[
 \frac{w_j(c_{ij}-c_i^*)}{w_i c_j^*}
 \le
 \left(\frac{2}{1-k}+\frac{2Bv}{1-k}\right)^\theta
 \frac{\min\{1,Bv\}}{v}.
\]
To maximize this expression, put $a_0=\frac{2}{1-k}$ and $b_0=\frac{2B}{1-k}$. For $v\le\frac1B$, it equals $B(a_0+b_0v)^\theta$, which is increasing. For $v\ge\frac1B$, it equals $\frac{(a_0+b_0v)^\theta}{v}$, whose logarithmic derivative is
\[
 \frac{\theta b_0}{a_0+b_0v}-\frac1v
 =\frac{(\theta-1)b_0v-a_0}{v(a_0+b_0v)}<0.
\]
Thus the maximum occurs at $v=\frac1B$ and equals $B\left(\frac{4}{1-k}\right)^\theta=C(k)$.

To complete \eqref{eq:power-pair-bound}, note that
\[
 \frac{d}{dk}\log C(k)
 =\frac{2\left(1+k-\log\frac{4}{1-k}\right)}{(1+k)^2}<0.
\]
Indeed, $\log\frac{4}{1-k}>1+k$ for $0\le k<1$. The limits of $C(k)$ at $0$ and $1$ are $4$ and $2$, respectively, so $2<C(k)\le4$.

If $W>0$, summing \eqref{eq:power-pair-bound} over $i,j$ and dividing by $W$ yields
\[
 \mathbb E[SC]
 =\frac1W\sum_{i,j}w_jc_{ij}
 \le\frac1W\left(W\sum_i c_i^*+C(k)W\sum_jc_j^*\right)
 =(1+C(k))OPT.
\]
If $W=0$, every agent's minimum feasible cost is its direct-route cost, and the fallback pathway is optimal.
\end{proof}

\subsection{A Sharper Guarantee for Zero Pathway Cost}

The proportional mechanism at $k=0$ admits a sharper analysis. Rather than partitioning agents by their positions relative to an optimal center, we compare the cost of each ordered pair of agents with their optimal costs. The inequality retains the direct-route cap and can then be summed over all pairs. We state it on the real line so that it can also be applied in Section~\ref{sec:improve}.

\begin{lemma}\label{lem:proportional-pair}
Let $z_1,\ldots,z_n,z^*\in\mathbb R$, and set
\[
 t_i=|z_i|,\qquad c_i^*=\min\{t_i,|z_i-z^*|\},\qquad
 c_{ij}=\min\{t_i,|z_i-z_j|\}.
\]
Then $t_jc_{ij}\le t_jc_i^*+2t_i c_j^*$ for every $i,j$.
\end{lemma}
\begin{proof}
If $c_i^*=t_i$, the first term on the right already bounds $t_jc_{ij}$. If $c_j^*=t_j$, the second term suffices because $c_{ij}\le t_i$. These cases include agents that retain their direct route at the comparison solution.

Otherwise, both agents strictly use the center $z^*$. Put $A=c_i^*=|z_i-z^*|$ and $D=c_j^*=|z_j-z^*|$. The triangle inequality gives
\[
 t_j=|z_j|\le|z_i|+|z_i-z^*|+|z_j-z^*|=t_i+A+D.
\]
Also, $|z_i-z_j|\le A+D$, so $c_{ij}\le\min\{t_i,A+D\}$. Subtracting $A$ retains both bounds:
\[
 c_{ij}-A\le\min\{t_i-A,D\}.
\]
If $c_{ij}\le A$, the claim is immediate. In the remaining case,
\[
 t_j(c_{ij}-A)
 \le(t_i+A+D)\min\{t_i-A,D\}\le2t_iD.
\]
For the last inequality, when $D\le t_i-A$ use $A+D\le t_i$; when $D\ge t_i-A$, expanding the difference between the right- and left-hand sides gives $(t_i+A)(D-(t_i-A))\ge0$.
\end{proof}

\begin{theorem}\label{thm:proportional-three}
For $k=0$, Mechanism~\ref{mec:power-proportional} is a randomized strategyproof $3$-approximation for social cost. The factor $3$ is tight for this mechanism as $n$ grows.
\end{theorem}
\begin{proof}
Strategyproofness follows from Theorem~\ref{thm:power-sp}. Map left-side locations to $z_i=x_i$ and right-side locations to $z_i=x_i-1$. Selecting $j$ gives exactly the cost $\min\{|z_i|,|z_i-z_j|\}$. By Lemma~\ref{lem:opt-basic}, an optimal pathway has one endpoint at an original facility; its costs therefore have the form $c_i^*$ in Lemma~\ref{lem:proportional-pair} for a feasible signed center $z^*$. If $T=\sum_i t_i>0$, summing that lemma gives
\[
 \mathbb E[SC]=\frac{1}{T}\sum_{i,j}t_jc_{ij}
 \le\frac{1}{T}\sum_{i,j}(t_jc_i^*+2t_ic_j^*)=3OPT.
\]
If $T=0$, all agents have zero cost.

For tightness, place $m$ agents at $a$ and one at $2a$, with $0<2a<o$. The optimal social cost is $a$. Selecting a report at $a$ gives social cost $a$, whereas selecting $2a$ gives $ma$. Their probabilities are $\frac{m}{m+2}$ and $\frac{2}{m+2}$, respectively. The ratio is $\frac{3m}{m+2}$, which tends to $3$.
\end{proof}

\begin{corollary}\label{cor:random-sc-upper}
There is a randomized strategyproof mechanism with social-cost approximation ratio at most
\begin{equation}\label{eq:random-sc-upper}
 U_n(k)=
 \begin{cases}
  \min\{3,n\},&k=0,\\
  \min\left\{1+C(k),\frac{n}{1+k(n-1)}\right\},&0<k<1.
 \end{cases}
\end{equation}
\end{corollary}
\begin{proof}
Choose between \textsc{PowerProportional} and \textsc{TwoExtreme} according to their proved bounds, using only the public parameters $n,k$. The choice is independent of the reported locations and hence preserves strategyproofness.
\end{proof}

\subsection{A Randomized Social-Cost Lower Bound}

Our lower bound uses two profiles with a common group of four agents, following the construction in \cite{cocoa2024mechanism}. We prove a pointwise certificate that allows an arbitrary nonnegative pathway fee. This lets us establish the lower bounds directly in the present model and in the real-line pathway model, without assuming that a mechanism must choose a pathway incident to a facility.

\begin{lemma}\label{lem:random-sc-certificate}
For $z\in\mathbb R$, $h\ge0$, and $t\ge0$, define $c_t=\min\{t,|t-z|+h\}$. Put
\[
 S_P=4c_{\frac{7}{10}}+3c_2,\qquad S_Q=4c_1+3c_2.
\]
Then
\begin{equation}\label{eq:random-sc-certificate}
 7S_P+11c_{\frac{7}{10}}\ge\frac{273}{10},
 \qquad
 11S_Q-11c_{\frac{7}{10}}\ge\frac{297}{10}.
\end{equation}
\end{lemma}
\begin{proof}
First take $h=0$. Both expressions are continuous and piecewise linear in $z$, with breakpoints among $0,\frac{7}{10},1,\frac75,2,4$, and are constant outside $[0,4]$. Their values at the breakpoints are
\[
\begin{array}{c|rrrrrr}
 z&0&\frac{7}{10}&1&\frac75&2&4\\ \hline
 7S_P+11c_{\frac{7}{10}}
 &\frac{693}{10}&\frac{273}{10}&\frac{327}{10}&\frac{399}{10}&\frac{273}{10}&\frac{693}{10}\\
 11S_Q-11c_{\frac{7}{10}}
 &\frac{1023}{10}&\frac{561}{10}&\frac{297}{10}&\frac{297}{10}&\frac{363}{10}&\frac{1023}{10}
\end{array}
\]
which proves the inequalities for $h=0$.

The first expression is nondecreasing in $h$. To handle the second, observe that if $c_{\frac{7}{10}}<\frac{7}{10}$, then $z>0$, $h<z$, and $z+h<\frac75$. It follows that agents at $1$ and $2$ also strictly use the pathway. On each linear piece where $c_{\frac{7}{10}}$ is not capped, the derivative of $S_Q-c_{\frac{7}{10}}$ with respect to $h$ is $7-1=6$. Once this cost is capped, its derivative is zero and $S_Q$ remains nondecreasing. Continuity handles the breakpoints, proving the inequalities for all $h\ge0$.
\end{proof}

\begin{theorem}\label{thm:random-sc-lower}
For $n\ge7$ and $0\le k<1$, every randomized strategyproof mechanism has a social-cost approximation ratio at least
\[
 \max\left\{1,\frac{285}{263+385k}\right\}.
\]
At $k=0$, this gives $\frac{285}{263}\approx1.08365$.
\end{theorem}
\begin{proof}
Fix $0<\lambda<\frac12$ and an obstacle with $2\lambda<o<1$. In profile $P$, place four agents at $\frac{7\lambda}{10}$ and three at $2\lambda$. In profile $Q$, move the first four agents to $\lambda$, keeping the other three fixed. Place all extra agents at $0$. For any output, left-side costs have the form
\[
 c_t=\min\{t,|t-a|+h\},\qquad h=k(b-a)+1-b\ge0.
\]
Apply Lemma~\ref{lem:random-sc-certificate} with center $\frac{a}{\lambda}$ and fee $\frac{h}{\lambda}$, and multiply its inequalities by $\lambda$. Taking expectations over the respective output distributions yields
\begin{align*}
 7\mathbb E[SC(P)]+11\mathbb E_P[c_{\frac{7\lambda}{10}}]
 &\ge\frac{273\lambda}{10},\\
 11\mathbb E[SC(Q)]-11\mathbb E_Q[c_{\frac{7\lambda}{10}}]
 &\ge\frac{297\lambda}{10}.
\end{align*}
The expectation version of partial group strategyproofness gives
$\mathbb E_P[c_{\frac{7\lambda}{10}}]\le\mathbb E_Q[c_{\frac{7\lambda}{10}}]$. The extra term in their sum is $11(\mathbb E_P[c_{\frac{7\lambda}{10}}]-\mathbb E_Q[c_{\frac{7\lambda}{10}}])\le0$. Removing this term therefore gives
\begin{equation}\label{eq:random-sc-weighted}
 7\mathbb E[SC(P)]+11\mathbb E[SC(Q)]\ge57\lambda.
\end{equation}
The feasible pathways $(2\lambda,1)$ and $(\lambda,1)$ give, respectively,
\begin{align*}
 OPT(P)&\le\frac{14\lambda}{5}+3k(1-2\lambda),\\
 OPT(Q)&\le3\lambda+7k(1-\lambda).
\end{align*}
Thus any approximation ratio $\rho$ satisfies
\[
 \rho\ge\frac{57\lambda}{\frac{263\lambda}{5}+k(98-119\lambda)}.
\]
Letting $\lambda\to\frac12$ proves $\rho\ge\frac{285}{263+385k}$. The trivial lower bound $1$ completes the statement. At $k=0$, the same ratio $\frac{285}{263}$ already follows for every fixed admissible $\lambda$.
\end{proof}

\section{Improved Bounds for the Real-Line Pathway Model} \label{sec:improve}

We consider the model of \cite{Chan023}, in which agents may be located anywhere on the real line and a facility is fixed at $0$. A mechanism selects a zero-cost pathway with endpoints $a,b\in\mathbb R$. The cost of an agent at $x$ is
\[
 c(x;a,b)=\min\{|x|,\,|x-a|+|b|,\,|x-b|+|a|\}.
\]
Reports are unrestricted locations on the real line.

There is a cost correspondence with our model when $k=0$: map a left-side location $x$ to $-x$, a right-side location $x$ to $1-x$, and a pathway $(a,b)$ to $(-a,1-b)$. Both original facilities map to $0$, and each agent's cost is preserved. This correspondence concerns a restricted set of locations and pathways; it does not identify the two mechanisms' feasible outputs or reporting domains. We therefore justify the results for the real-line model directly. For randomized mechanisms, SP and approximation are defined using expected costs as above, but deviations may now be arbitrary real locations.

For any pathway, label its endpoints so that $|a|\le |b|$ and write $h=|a|$. The inequality $|x-a|+|b|\ge |x|$ gives
\begin{equation}\label{eq:real-line-cost}
 c(x;a,b)=\min\{|x|,\,|x-b|+h\}.
\end{equation}
This is only a relabeling of the actual output: we retain the nonnegative term $h$ throughout the lower-bound proofs. In particular, if an agent at a positive location strictly benefits from the pathway, then $b>0$.

For optimization, replacing $(a,b)$ by $(0,b)$ weakly decreases every agent's cost. Thus the optimal value equals that in the model with one fixed facility at $0$ and one additional facility. A deterministic or randomized strategyproof mechanism for that model also remains strategyproof when its additional facility $b$ is implemented as the pathway $(0,b)$. Consequently, the deterministic maximum-cost upper bound of $2$ in \cite{cocoa2024mechanism} applies here. We match it by improving the deterministic lower bound of $\frac{3}{2}$ in \cite{Chan023} to $2$. This establishes a tight deterministic bound. This upper-bound transfer does not justify transferring lower bounds: replacing the outputs of a strategyproof mechanism by pathways incident to $0$ need not preserve strategyproofness. Our lower-bound proofs therefore allow arbitrary endpoints.

\subsection{Deterministic Mechanisms}

\begin{theorem}\label{thm:improve-1}
For every $n\ge2$, no deterministic strategyproof mechanism has a maximum-cost approximation ratio strictly smaller than $2$ in the real-line model.
\end{theorem}

\begin{proof}
Suppose that such a mechanism has approximation ratio $\rho=2-\delta$, where $0<\delta\le1$. Fix $L>2$, put $r=L+2$, and consider profiles with one agent at $s\in[L,r)$ and all other agents at $r$. The pathway $\left(0,\frac{s+r}{2}\right)$ gives
\[
 OPT\le\frac{r-s}{2},
 \qquad ALG\le\frac{\rho(r-s)}{2}<r-s\le2<L.
\]
Every agent therefore strictly prefers the pathway to the direct route. Write its cost as in \eqref{eq:real-line-cost}, with output parameters $b,h$. The approximation guarantee implies
\[
 \max\{|s-b|,|r-b|\}+h\le\frac{\rho(r-s)}{2},
\]
and hence
\begin{equation}\label{eq:real-mc-endpoints}
 s+\frac{\delta(r-s)}{2}+h
 \le b\le
 r-\frac{\delta(r-s)}{2}-h.
\end{equation}
In particular, $s<b<r$.

Set $s_1=L$. If the output at step $j$ has parameters $b_j,h_j$, define $s_{j+1}=b_j$, changing only the first agent's report. By \eqref{eq:real-mc-endpoints}, this defines another profile in the same interval and gives $s_j<b_j<b_{j+1}<r$. Strategyproofness in both directions yields
\begin{align*}
 b_j-s_j+h_j
 &\le c(s_j;a_{j+1},b_{j+1})
 \le b_{j+1}-s_j+h_{j+1},\\
 b_{j+1}-b_j+h_{j+1}
 &\le c(b_j;a_j,b_j)\le h_j.
\end{align*}
The truthful costs on the left use the pathway, while the costs under deviations are bounded above by the corresponding pathway routes. Thus these inequalities remain valid even if a deviating agent uses the direct route. Together they imply
\[
 b_{j+1}+h_{j+1}=b_j+h_j=b_1+h_1.
\]
At the initial profile, \eqref{eq:real-mc-endpoints} gives $b_1+h_1\le r-\delta$. On the other hand, the same inequality at step $j+1$ gives
\[
 r-b_{j+1}\le\left(1-\frac{\delta}{2}\right)(r-b_j).
\]
Choose a finite $J$ such that
\[
 2\left(1-\frac{\delta}{2}\right)^{J-1}<\delta.
\]
Then $b_J>r-\delta$, whereas the invariant and $h_J\ge0$ imply $b_J\le b_J+h_J=b_1+h_1\le r-\delta$, a contradiction.
\end{proof}

For social cost, \cite{Chan023} gives a lower bound of $\frac{3}{2}$ for deterministic strategyproof mechanisms and an upper bound of $n$. We first improve the strategyproof lower bound to $2$, and then establish an additional lower bound of $n-1$ under group strategyproofness.

\begin{theorem}\label{thm:real-sc-sp}
For every $n\ge2$, no deterministic strategyproof mechanism has a social-cost approximation ratio strictly smaller than $2$ in the real-line model.
\end{theorem}

\begin{proof}
Suppose that such a mechanism has approximation ratio $1\le\rho<2$. Place one agent at $-1$, one at $1$, and all remaining agents at $0$. The pathway $(0,1)$ gives $OPT\le1$, so $ALG\le\rho<2$. Thus at least one of the two nonzero agents strictly benefits from the pathway. By \eqref{eq:real-line-cost}, agents on opposite sides of $0$ cannot both strictly benefit: if $x>0$ benefits, then $b>0$, and every negative location retains its direct-route cost; the other case is symmetric.

Let $\sigma\in\{-1,1\}$ be the location of the agent that benefits. Keep that agent fixed, and move the other agent from $-\sigma$ to $-\sigma t$, where $t\in[1,\rho+1]$. Write $c(t)$ for the moving agent's truthful cost and $w(t)=t-c(t)$ for its improvement over the direct route. The cost of every fixed pathway is $1$-Lipschitz in the agent's location. Therefore the proof of Lemma~\ref{lem:truthful-cost-continuity}, applied to this single agent, shows that $c$ and $w$ are continuous. Initially $w(1)=0$.

At every such profile, the pathway $(0,-\sigma t)$ gives $OPT\le1$. If $w(t)>0$, the fixed agent at $\sigma$ must incur its direct cost $1$. The approximation guarantee then implies
\[
 c(t)+1\le\rho,
 \qquad w(t)=t-c(t)\ge t-\rho+1\ge2-\rho>0.
\]
Thus $w$ cannot take a value strictly between $0$ and $2-\rho$. However, at $t=\rho+1$, the same approximation guarantee gives $c(t)\le\rho$, hence $w(\rho+1)\ge1$. This contradicts continuity and $w(1)=0$.
\end{proof}

\begin{theorem}\label{thm:improve-2}
For every $n\ge2$, no deterministic group-strategyproof mechanism has a social-cost approximation ratio strictly smaller than $\max\{2,n-1\}$ in the real-line model.
\end{theorem}

\begin{proof}
The lower bound of $2$ follows from Theorem~\ref{thm:real-sc-sp}. It remains to prove the bound of $n-1$ for $n\ge4$. Let $p=n-1\ge3$ and suppose that a group-strategyproof mechanism has approximation ratio $1\le\rho<p$. Fix $r>0$ and put
\[
 \epsilon=\frac{r}{4\rho p},\qquad \alpha=\frac{\epsilon}{2}.
\]
Consider profiles with $p$ agents at $s\in[\epsilon,r]$ and one agent at $r$. Write $g(s)$ for the group's common truthful cost. Lemmas~\ref{lemma:partial-gsp} and~\ref{lem:truthful-cost-continuity} also apply on the real line: their proofs use only strategyproofness, identical costs at identical locations, and the fact that the cost of each fixed pathway is $1$-Lipschitz. In particular, $g$ is continuous.

At $s=r$, the pathway $(0,r)$ gives zero social cost, so $g(r)=0$. At $s=\epsilon$, the same pathway gives $OPT\le p\epsilon$, and hence $ALG\le\rho p\epsilon=\frac{r}{4}$. The agent at $r$ must use the pathway. In the notation of \eqref{eq:real-line-cost}, this implies
\[
 |r-b|+h\le\frac{r}{4},\qquad b\ge\frac{3r}{4}.
\]
Since $\epsilon\le\frac{r}{8}$, every agent at $\epsilon$ uses the direct route, giving $g(\epsilon)=\epsilon$. By continuity, there is an $s\in(\epsilon,r)$ with $g(s)=\alpha<s$.

Fix this profile and its output $(a,b)$, labeled as in \eqref{eq:real-line-cost}. The group strictly benefits from the pathway. The remaining agent can report $s$, in which case all reports coincide at $s>0$ and a finite approximation ratio forces the pathway $(0,s)$. Indeed, zero cost at $s>0$ requires $|s-b|+h=0$, hence $b=s$ and $h=|a|=0$. Strategyproofness therefore bounds its truthful cost by $r-s<r$. Thus all agents strictly prefer the pathway, and their costs are $|x_i-b|+h$.

We claim that $h=0$. If $h>0$, all agents could jointly report $b$. Here $b>0$, since agents at positive locations strictly benefit from the original pathway. At the common report $b$, the zero optimum forces the pathway $(0,b)$. Every agent's cost would then be at most $|x_i-b|$, strictly below its original cost $|x_i-b|+h$. This contradicts group strategyproofness, proving the claim. It follows that $|s-b|=\alpha$.

Now change only the remaining agent's location from $r$ to $b$. By reporting $r$ at the new profile, this agent could obtain the old pathway $(0,b)$ and incur zero cost. Strategyproofness forces its new truthful cost to be zero as well, so the new pathway must have endpoints $0$ and $b$. Each group member then incurs cost $\min\{s,|s-b|\}=\alpha$, and \(ALG=p\alpha\).
The feasible pathway $(0,s)$ gives $OPT\le|b-s|=\alpha$. Therefore $ALG\ge p\alpha>\rho\,OPT$, a contradiction.
\end{proof}

% \paragraph{The strategyproof case.}
For $n\ge4$, the $n-1$ lower bound for arbitrary deterministic strategyproof mechanisms remains unproved here; Theorem~\ref{thm:real-sc-sp} establishes a lower bound of $2$, which also proves the $n-1$ bound when $n=3$. The preceding proof uses group strategyproofness exactly once, to eliminate the positive term $h$ when all agents use the pathway. The deviating coalition in that step contains agents at two different locations, so partial group strategyproofness cannot justify it. Nor does the existence of an optimal pathway of the form $(0,b)$ imply that a strategyproof mechanism must select such a pathway. If a mechanism is required to choose a pathway incident to $0$ on every profile, then $h=0$ holds directly and the same proof establishes the $n-1$ lower bound under strategyproofness alone.

\subsection{Randomized Mechanisms}

For maximum cost, the randomized lower bound of $\frac32$ was already established in \cite{Chan023}; we retain it. The deterministic $2$-approximation discussed above also gives a randomized upper bound of $2$. For social cost, we sharpen the analysis of the same proportional mechanism used in \cite{Chan023}, reducing its guarantee from $6$ to $3$, and strengthen the randomized lower bound.

\begin{theorem}\label{thm:real-random-sc-upper}
In the real-line model, the mechanism that selects agent $j$ with probability $\frac{|x_j|}{\sum_i|x_i|}$ and outputs $(0,x_j)$ is strategyproof in expectation and has social-cost approximation ratio $3$. This factor is tight for this mechanism as $n$ grows. If all reports are zero, it outputs $(0,0)$.
\end{theorem}
\begin{proof}
Fix a true location $x$ and put $t=|x|$. Let
\[
 S=\sum_{j\ne i}|x_j|,\qquad
 A=\sum_{j\ne i}|x_j|\min\{t,|x-x_j|\}.
\]
Then $0\le A\le St$, and the truthful expected cost is $\frac{A}{S+t}$ whenever $S+t>0$. An arbitrary real-line report $r$, including one of the opposite sign, gives expected cost
\[
 \frac{A+|r|\min\{t,|x-r|\}}{S+|r|}.
\]
Although the probabilities of selecting other agents change, their locations and unnormalized weights remain fixed. This is why the same numerator contribution $A$ appears in both expressions. When both denominators are positive, put $u=|r|$ and $d_r=\min\{t,|x-r|\}$. The deviating expected cost minus the truthful one equals
\[
 \frac{u(S+t)d_r-(u-t)A}{(S+u)(S+t)}.
\]
If $u\le t$, the numerator is nonnegative. If $u>t$, the reverse triangle inequality gives $|x-r|\ge u-t$, and therefore
\[
 ud_r\ge u\min\{t,u-t\}\ge t(u-t)
 \ge\frac{A(u-t)}{S+t}.
\]
For the middle inequality, if $u-t\le t$ use $u\ge t$; otherwise use $u\ge u-t$. This again makes the numerator of the cost difference nonnegative. The argument depends only on absolute values and thus also covers reports on the opposite side of $0$. If $S+u=0$, the fallback gives true cost $t$, which cannot improve on truthfulness; if $S+t=0$, the agent is truly at $0$ and always has zero cost.

For approximation, an optimal pathway may be chosen as $(0,z^*)$. Lemma~\ref{lem:proportional-pair}, with $z_i=x_i$, and summation as in Theorem~\ref{thm:proportional-three} give $\mathbb E[SC]\le3OPT$. The instance with $m$ agents at $a>0$ and one at $2a$ again has ratio $\frac{3m}{m+2}\to3$.
\end{proof}

\begin{theorem}\label{thm:real-random-sc-lower}
For $n\ge7$, every randomized strategyproof mechanism in the real-line model has social-cost approximation ratio at least $\frac{285}{263}$.
\end{theorem}
\begin{proof}
Use profiles $P=(\frac{7}{10},\frac{7}{10},\frac{7}{10},\frac{7}{10},2,2,2)$ and $Q=(1,1,1,1,2,2,2)$, padding with agents at $0$ if necessary. For every realized pathway, \eqref{eq:real-line-cost} expresses each positive-location cost as $\min\{t,|t-b|+h\}$ with $h\ge0$. Lemma~\ref{lem:random-sc-certificate} therefore applies to arbitrary endpoints, including pathways not incident to $0$. Partial group strategyproofness in expectation yields, exactly as in \eqref{eq:random-sc-weighted},
\[
 7\mathbb E[SC(P)]+11\mathbb E[SC(Q)]\ge57.
\]
Since $OPT(P)=\frac{14}{5}$ and $OPT(Q)=3$, it follows that
\[
 \rho\ge\frac{57}{7\cdot\frac{14}{5}+11\cdot3}
 =\frac{285}{263}.
\]
\end{proof}

Theorem~\ref{thm:real-random-sc-lower} raises the randomized social-cost lower bound of Chan and Wang~\cite{Chan023} from $1.02$ to $\frac{285}{263}$. The proof applies to arbitrary pathway endpoints through the nonnegative fee term $h$, rather than transferring a lower bound from the fixed-facility model. The upper bound $3$ is tight for the proportional mechanism, whereas the optimal ratio over all randomized SP mechanisms remains between $\frac{285}{263}$ and $3$ for $n\ge7$.

\section{Conclusion}

We studied deterministic and randomized strategyproof mechanisms for constructing a pathway between two regions separated by an obstacle. For deterministic mechanisms, the maximum-cost approximation ratio is exactly $\frac{2}{1+k}$, and the social-cost bounds are linear in the population size when $k=0$. Under strategyproofness in expectation, the power-proportional mechanism gives a social-cost guarantee independent of both $n$ and $k$, with a sharper factor $3$ when the pathway has zero traversal cost. We also obtained randomized lower bounds for both objectives.

For the real-line pathway model of \cite{Chan023}, our deterministic maximum-cost lower bound of $2$ matches the upper bound obtained from \cite{cocoa2024mechanism}. We also strengthen the deterministic social-cost lower bounds under SP and GSP. For randomized social cost, we sharpen the guarantee of Chan and Wang's proportional mechanism from $6$ to $3$ and raise their lower bound from $1.02$ to $\frac{285}{263}$ for $n\ge7$.

Several gaps remain. For randomized maximum cost in the obstacle model, improving the upper bound $\frac{2}{1+k}$ or strengthening the lower bound $\frac{3+2k}{2+3k}$ is open. For randomized social cost, a sharper analysis of power-proportional sampling may improve the positive-$k$ guarantee; stronger lower bounds are also needed, especially away from $k=0$. Finally, the $n-1$ social-cost lower bound for arbitrary deterministic SP mechanisms in the real-line model remains unresolved for $n\ge4$.

\bibliographystyle{plain}
\bibliography{myreferences}

\end{document}